\documentclass[11pt]{article} 
\usepackage[margin=1in]{geometry}

\usepackage{float}

\usepackage{framed}

\usepackage[linesnumbered,boxruled,vlined]{algorithm2e}
\usepackage{verbatim}

\usepackage{amssymb,amsfonts,amsmath,amsthm}

\usepackage{enumerate}

\usepackage[font=small,labelfont=bf]{caption}

\usepackage{color}

\usepackage{graphicx}
\usepackage{tabulary}
\usepackage{diagbox}
\usepackage[dvipsnames]{xcolor}
\usepackage{pgfplots}
\usepackage[colorlinks=true, linkcolor=blue, urlcolor=blue, citecolor=ForestGreen]{hyperref}

\usepackage{comment}

\newcommand{\delete}{\texttt{Delete}}

\newcommand{\cost}{\operatorname{\textsc{Cost}}}

\newcommand{\ins}{\ttx{Insert}}

\usepackage{setspace}

\newcommand{\poly}{\operatorname{poly}}
\newcommand{\polylog}{\operatorname{polylog}}

\newcommand{\LFMIS}{\operatorname{LFMIS}} 
\newcommand{\elim}{\operatorname{elim}} 
\newcommand{\OPT}{\operatorname{OPT}} 

\newcommand{\ALG}{\mathcal{L}_{k+1}}

\newcommand{\R}{\mathbb{R}}

\newcommand{\ex}[1]{\mathop{{\bf E}\left[ #1 \right]}}
\newcommand{\exx}[2]{\mathop{{\bf E}}_{#1}\left[ #2 \right]}

\newcommand{\pr}[1]{\operatorname{{\bf Pr}}\left[ #1 \right]}
\newcommand{\prb}[2]{\mathop{{\bf Pr}}_{#1}\left[ #2 \right]}

\newcommand{\Run}{\operatorname{Time}}

\newcommand{\bD}{\mathbf{D}}

\newcommand{\bX}{\mathbf{X}}

\newcommand{\bZ}{\mathbf{Z}}

\usepackage{footnote}

\newcommand{\cA}{\mathcal{A}}
\newcommand{\cB}{\mathcal{B}}
\newcommand{\cC}{\mathcal{C}}
\newcommand{\cD}{\mathcal{D}}
\newcommand{\cE}{\mathcal{E}}
\newcommand{\cF}{\mathcal{F}}

\newcommand{\cH}{\mathcal{H}}
\newcommand{\cI}{\mathcal{I}}

\newcommand{\cL}{\mathcal{L}}

\newcommand{\cP}{\mathcal{P}}
\newcommand{\cQ}{\mathcal{Q}}

\newcommand{\cX}{\mathcal{X}}

\newcommand{\eps}{\epsilon}
\newcommand{\ttx}[1]{\texttt{#1}}

\newtheorem{theorem}{Theorem}
\newtheorem{lemma}{Lemma}
\newtheorem{corollary}{Corollary}

\newtheorem{proposition}{Proposition}
\newtheorem{definition}{Definition}

\newtheorem{claim}{Claim}

\usepackage[framemethod=TikZ]{mdframed}
\newcounter{Frame}
\newenvironment{Frame}[1][h]{%
	\refstepcounter{Frame}
	\begin{mdframed}[%
		frametitle={#1},
		skipabove=\baselineskip plus 2pt minus 1pt,
		skipbelow=\baselineskip plus 2pt minus 1pt,
		linewidth=1.0pt,
		frametitlerule=true,
		]%
	}{%
	\end{mdframed}
}

\title{Optimal Fully Dynamic $k$-Centers Clustering} 
\author{
	MohammadHossein Bateni\\
	Google Research \\
	\texttt{bateni@google.com}
	\and
Hossein Esfandiari \\
Google Research \\
	\texttt{esfandiari@google.com}
	\and 
	Rajesh Jayaram \\
	Google Research \\
	\texttt{rkjayaram@google.com} 
	\and
	Vahab Mirrokni \\
	Google Research \\
	\texttt{mirrokni@google.com}
}

\date{}

\begin{document}
\maketitle


		\begin{abstract}
		
		
		We present the first  algorithm for fully dynamic $k$-centers clustering in an arbitrary metric space that maintains an optimal $2+\eps$ approximation in $O(k \cdot  \polylog(n,\Delta))$ amortized update time. 
		Here, $n$ is an upper bound on the number of active points at any time, and $\Delta$ is the aspect ratio of the data. Previously, the best known amortized update time was $O(k^2\cdot \polylog(n,\Delta))$, and is due to Chan, Gourqin, and Sozio \cite{chan2018fully}. We demonstrate that the runtime of our algorithm is optimal up to $\polylog(n,\Delta)$ factors, even for insertion-only streams, which closes the complexity of fully dynamic $k$-centers clustering. In particular, we prove that any algorithm for $k$-clustering tasks in arbitrary metric spaces, including $k$-means, $k$-medians, and $k$-centers, must make at least $\Omega(n k)$ distance queries to achieve any non-trivial approximation factor. 
		
		Despite the lower bound for arbitrary metrics, we demonstrate that an update time sublinear in $k$ is possible for metric spaces which admit locally sensitive hash functions (LSH). Namely, we demonstrate a black-box transformation which takes a locally sensitive hash family for a metric space and produces a faster fully dynamic $k$-centers algorithm for that space. In particular, for a large class of metrics including Euclidean space, $\ell_p$ spaces, the Hamming Metric, and the Jaccard Metric, for any $c > 1$, our results yield a $c(4+\eps)$ approximate $k$-centers solution in $O(n^{1/c} \cdot \polylog(n,\Delta))$ amortized update time, simultaneously for all $k \geq 1$. Previously, the only known comparable result was a $O(c \log n)$ approximation for Euclidean space due to Schmidt and Sohler, running in the same amortized update time \cite{schmidt2019fully}.
		


	\end{abstract}
	
\thispagestyle{empty}
\newpage
\parskip 4.0 pt
\tableofcontents
\thispagestyle{empty}
\newpage
\parskip 7.2pt 
\pagenumbering{arabic}


	\section{Introduction}
	Clustering is a fundamental and well-studied problem in computer science, which arises in approximation algorithms, unsupervised learning, computational geometry, classification, community detection, image segmentation, databases, and other areas  \cite{hansen1997cluster,schaeffer2007graph,fortunato2010community,shi2000normalized,arthur2006k,tan2013data,coates2012learning}. 
	The goal of clustering is to find a structure in data by grouping together similar data points. Clustering algorithms optimize a given objective function which characterizes the quality of a clustering. One of the classical and best studied clustering objectives is the $k$-centers objective. 
	
	Specifically, given a metric space $(\cX,d)$ and a set of points $P \subseteq \cX$, the goal of $k$-centers clustering is to output a set $C \subset \cX$ of at most $k$ ``centers'', such that the maximum distance of any point $p \in P$ to the nearest center $c \in C$ is minimized. In other words, the goal is to minimize the objective function 
	$$\max_{p \in P}d(p,C)$$ 
	where $d(p,C) = \min_{c \in C} d(p,c)$. The $k$-centers clustering problem admits several well-known greedy $2$-approximation algorithms \cite{gonzalez1985clustering,hochbaum1986unified}. However, it is known to be NP-Hard to approximate the objective to within a factor of $(2-\epsilon)$ for any constant $\eps>0$~\cite{hsu1979easy}. Moreover, even restricted to Euclidean space, it is still NP-Hard to approximate beyond a factor of $1.822$ \cite{feder1988optimal,bern1997approximation}.
	
	While the approximability of many clustering tasks, including $k$-centers clustering, is fairly well understood in the static setting, the same is not true for \textit{dynamic datasets}. Recently, due to the proliferation of data and the rise of modern computational paradigms where data is constantly changing, 
	there has been significant interest in developing dynamic clustering algorithms~\cite{cohen2016diameter,lattanzi2017consistent,chan2018fully,schmidt2019fully,goranci2019fully,henzinger2020fully,henzinger2020dynamic,fichtenberger2021consistent}. In the incremental dynamic setting, the dataset $P$ is observed via a sequence of insertions of data points, and the goal is to maintain a good $k$-centers clustering of the current set of active points. In the \textit{fully dynamic} setting, points can be both inserted and deleted from $P$.
	
	The study of dynamic algorithms for $k$-centers was initated by Charikar, Chekuri, Feder, and Motwani~\cite{charikar2004incremental}, whose ``doubling algorithm'' maintains an $8$-approximation in amortized $O(k)$ update time. However, the doubling algorithm is unable to handle deletions of data points. It was not until recently that the first fully dynamic algorithm for $k$-centers, with update time better than naively recomputing a solution from scratch, was developed. In particular, the work of Chan, Guerqin, and Sozio \cite{chan2018fully} proposed a randomized algorithm that maintains an optimal $(2+\eps)$-approximation in   $O(\frac{\log \Delta}{\eps}k^2 )$ amortized time per update, where $\Delta$ is the aspect ratio of the dataset. Since then, algorithms with improved runtime have been demonstrated for the special cases of Euclidean space \cite{schmidt2019fully} (albeit, with a larger approximation factor), and for spaces with bounded doubling dimension \cite{goranci2019fully}. However, despite this progress, to date no improvements have been made for fully dynamic $k$-centers in general metric spaces beyond the quadratic-in-$k$ amortized runtime of \cite{chan2018fully}.


	\subsection{Our Contributions}
		In this work, we resolve the complexity of fully dynamic $k$-centers clustering up to polylogarithmic factors in the update time. Specifically, we give an algorithm with nearly linear in $k$ amortized update time and optimal $(2+\eps)$ approximation.
	In what follows, let $\Delta$ denote the aspect ratio of the data,\footnote{Namely, for any two points $x,y$ in the active data set at any point, we have $r_{\min} \leq d(x,y) \leq r_{\max}$, and $\Delta = r_{\max}/r_{\min}$.} and let $n$ be an upper bound on the number of active points at any given time. Our main result is as follows.  
	
	\begin{theorem}\label{thm:main}
	There is a fully dynamic algorithm that, on a sequence of insertions and deletions of points from an arbitrary metric space $(\cX,d)$, maintains a $(2+\eps)$-approximation to the optimal $k$-centers clustering. The amortized update time of the algorithm is $O(\frac{\log \Delta \log n}{\eps}(k  + \log n))$ in expectation, and $O(\frac{\log \Delta \log n}{\eps}(k  + \log n) \log \delta^{-1})$  with probability $1-\delta$ for any $\delta \in (0,\frac{1}{2})$.
	\end{theorem}
		
	The algorithm of Theorem \ref{thm:main} is randomized, and works in the standard \textit{oblivious adversary} model.\footnote{In the oblivious adversary model, the sequence of updates is worst case, but fixed in advance. In other words, one can think of the stream as being chosen by an adversary that knows the algorithm being used, but which is not aware of the random bits being used by the algorithm.} Moreover, the algorithm does not need to know $n$ or the length of the stream in advance. 
	In Section  \ref{sec:LB}, we demonstrate that any dynamic algorithm for $k$-centers that achieves a non-trivial approximation factor must run in amortized time $\Omega(k)$, even when points are only inserted (see Theorem \ref{thm:LBMain} below). These two results together resolve the amortized complexity of fully dynamic $k$-centers clustering up to logarithmic factors.

	The main technical contribution of our work is a new fully dynamic algorithm for maintaining either a maximal independent set of size at most $k$, or an independent set of size at least $k+1$, in a graph which receives a fully dynamic sequence of \textit{vertex} insertions and deletions, in amortized update time $O(k \log n + \log^2 n)$ 
	(Theorem \ref{thm:LFMISMain}). By a well-known reduction of Hochbaum and Shmoys \cite{hochbaum1986unified}, the latter is sufficient to obtain a $(2+\eps)$ approximation to $k$-centers. The algorithm, which is presented in Section \ref{sec:generalMetric}, is an extension of fully dynamic MIS algorithms to vertex-valued updates, and may be of independent interest. See Section \ref{sec:tech} for a further overview of the algorithm. 
	

\paragraph{Deterministic Fully Dynamic $k$-Centers.} 
To complement the randomized algorithm  of Theorem \ref{thm:main}, we demonstrate the existence of a \textit{deterministic} algorithm for fully dynamic $k$-centers which runs in $O(k \cdot \polylog(\Delta,n))$ amortized pdate time, albiet at the slightly larger cost of a $O(\log n)$ approximation. Specifically, in Section \ref{sec:deterministic}, we prove the following theorem.

\noindent \textbf{Theorem} \ref{thm:Deterministic} . {\it
There is a deterministic algorithm that, on a sequence of insertions and deletions of points from an arbitrary metric space $(\cX,d)$, maintains a $(2+\eps)\lceil \log (n(1+\eps)) \rceil$-approximation to the optimal $k$-centers clustering, where $\eps>0$ is any constant. The amortized update time of the algorithm is $O(k \log \Delta \log n \log k   )$, and the worst-case update time of an insertion is $O(k \log \Delta \log n  \log k)$, and the worst-case update time of a deletion is $O(k^2 \log \Delta \log n   \log k )$. 
\it}

The algorithm of Theorem \ref{thm:Deterministic}, while obtaining a sub-optimal approximation factor, enjoys many of the benefits of deterministic algorithms, such as never failing to run within the stated amortized bounds, even against an adaptively chosen sequence of updates (as opposed to an obliviously chosen sequence of updates). Moreover, the update time of any deletion is always at most  and $O(k^2 \cdot \polylog (\Delta,n))$ in the worst case, whereas other algorithms for $(2+\eps)$ approximate fully dynamic $k$-centers, such as the algorithm of \cite{chan2018fully}, run in $\Omega(n)$ time in the worst case. 

\paragraph{Query Complexity Lower Bounds for Arbitrary Metric Spaces.} 
While for specialized metric spaces such as Euclidean space, it was previously known that one can approximate the $k$-centers solution in sublinear in $k$ update time (albeit with a larger approximation factor), on the other hand, for general metric spaces no such result was known. The challenge of this task is that clustering algorithms for general metric spaces do not know the metric in advance; namely, the metric itself, in addition to the dataset, is part of the input. Therefore, a $k$-centers algorithm for arbitrary metric spaces must garner information about the metric via computing distances between points in the dataset.

In other words, the input to a general (offline) metric space clustering algorithm can be described by the distance matrix $\bD \in \R^{n \times n}$, corresponding to an arbitrary metric over a dataset of $n$ points. We note that such an algorithm may not need to read the entire matrix $\bD$ --- for instance, the algorithm of Theorem \ref{thm:main} reads at most $\tilde{O}(nk)$ entries. Previously, it was shown by Mettu and Plaxton  that any algorithm that achieves an $O(1)$-approximation for $k$-medians and $k$-means must query at least $\Omega(nk)$ entries of the distance matrix \cite{mettu2004optimal,mettu2002approximation}. However, this does not rule out the possibility of a $o(nk)$ query algorithm with a constant approximation for $k$-centers, or a $o(nk)$-query algorithm with super-constant approximation for $k$-medians or $k$-means. To date, to the best of the authors' knowledge, both of the latter possibilities remained open. 

In Section \ref{sec:LB}, we demonstrate that this query complexity, and therefore runtime, of Theorem \ref{thm:main} is tight up to logarithmic factors. In fact, we demonstrate that any algorithm that returns a non-trivial approximation (less than the aspect ratio $\Delta$) to the optimal clustering objective must make $\Omega(nk)$ queries to the distance matrix $\bD$. This bound holds even for the more general $(k,z)$-clustering problem, including the well-studied $k$-medians and $k$-means (for $z=1,2$ respectively), which is to minimize the sum of all $z$-th powers of distances from points to their respective centers, as well as the $k$-centers objective. This extends the lower bound of \cite{mettu2004optimal}, demonstrating that no approximation whatsoever is possible for $k$-clustering tasks in general metric spaces using $o(nk)$ queries.

In the following, for $z > 0$, we write $\OPT_{k,z}(P)$ to denote the optimal cost of $(k,z)$-clustering on a point set $P$, and write $\OPT_{k,\infty}(P)$ to denote the optimal $k$-centers cost. We prove the following:

\begin{theorem}\label{thm:LBMain}
Fix any $k\geq 1,$ $n > k$, and let $R > 1$ be any arbitrary approximation factor (possibly depending on $k,n$). Let $z >0$ be either a constant, or set $z = \infty$.  Let $\cA$ be any algorithm that, given as input a distance matrix $\bD \in \R^{n \times n}$ corresponding to an arbitrary metric over a set $P$ of $n$ points, correctly distinguishes with probability at least $2/3$ whether $\OPT_{k,z}(P) \leq 1$ or $\OPT_{k,z}(P) \geq R$. Then $\cA$ must make at least $\Omega(T)$ queries in expectation to $\bD$, where
\[T = \begin{cases}
nk & \text{ if } n = \Omega(k \log k) \\
k^2 & \text{ otherwise.} \\
\end{cases}\]
\end{theorem}

Note that the above query complexity lower bound holds even for offline algorithms, which are given $\bD$ all at once and allowed to make an arbitrary sequence of adaptive queries. In particular, this demonstrates an $\Omega(k)$ lower bound on the expected amortized update time of any dynamic $k$-centers algorithm, even in the \textit{insertion-only} model, so long as either $n = \Omega(k \log k)$ or $n = O(k^2)$ (and otherwise the lower bound is at least $\Omega(k /\log k)$). Thus, taken together, Theorems \ref{thm:main} and \ref{thm:LBMain} resolve the amortized complexity of fully dynamic $k$-centers clustering in arbitrary metric spaces up to logarithmic factors. 

\renewcommand{\arraystretch}{1.3}
\begin{figure}[t]
\begin{tabular}{|c|c|c|c|c|}\hline

	Metric Space & \textbf{Our Approx.} &\textbf{ Our Runtime} & Prior Approx. &  Prior Runtime \\ \hline
	
	Arbitrary Metric Space& $2+\eps $& $\tilde{O}(  k) $& $2+\eps$ &  $ \tilde{O}( k^2) $ \cite{chan2018fully}\\ \hline
	$(\R^d,\ell_p)$, $p \in [1,2]$ & $c (4+\eps) $ &$\tilde{O}(  n^{1/c}) $& $O(c \cdot \log n)$ &  $\tilde{O}( n^{1/c}) $ \cite{schmidt2019fully}\\ \hline
	Eucledian Space $(\R^d,\ell_2)$  & $c (\sqrt{8}+\eps)  $ & $\tilde{O}(  n^{1/c^2 + o(1)}) $& $O(c \cdot \log n)$ & $ \tilde{O}( n^{1/c}) $ \cite{schmidt2019fully}\\ \hline
	Hamming Metric &  $c (4+\eps) $ &$\tilde{O}(  n^{1/c}) $& -- &  -- \\ \hline
	Jaccard Metric &  $c (4+\eps)  $ &$\tilde{O}(  n^{1/c}) $& -- &  -- \\ \hline
    \begin{tabular}{c}
          EMD over $[D]^d$   \\
          $d= O(1)$
     \end{tabular} &  $O(\log D \cdot   c )  $ &$\tilde{O}(  n^{1/c}) $& -- &  -- \\ \hline
     \begin{tabular}{c}
          EMD over $[D]^d$   \\
          sparsity $s$
     \end{tabular} &  $O(c \log s n \log d )  $ &$\tilde{O}(  n^{1/c}) $& -- &  -- \\ \hline
\end{tabular}
\caption{
Our approximation and amortized runtimes as compared to the previously best known results. Factors of $\log \Delta,\log M,d,\frac{1}{\eps}$ are omitted from the table, although we remark that the dependency on $d,1/\eps$ is linear, as in previous works, and the dependency on $\log \Delta$ is nearly linear $\tilde{O}(\log \Delta)$ for the LSH-based algorithms, and linear for arbitrary metric spaces (also as in previous works). Furthermore, there is no $\log M$ dependency in the runtime of the arbitrary metric space algorithm.  Note that the algorithms for LSH-spaces maintain a $k$-centers solution simultaneously for all $k \leq n$. 
Furthermore, the first result on Earth Mover's Distance (EMD) over constant dimensional Euclidean space $[D]^d$ follows immediately from embedding EMD into $\ell_1$ with distortion $O(\log D)$ of Indyk and Thaper, followed by an application of our algorithm for $\ell_1$ (see \cite{IT03} for the embedding and further definitions). The second result on EMD for subsets of $[D]^d$ with size at most $s$ follows similarly via the embedding into $\ell_1$ of Andoni, Indyk, and Krauthgamer \cite{andoni2008earth}.}
\label{fig:results}
\end{figure}

\paragraph{Improved Fully Dynamic $k$-centers via Locally Sensitive Hashing.}
The lower bound of Theorem \ref{thm:LBMain} demonstrates that, in general, one cannot beat $\Omega(k)$ amortized update time for fully Dynamic $k$-centers clustering. However, as observed in \cite{schmidt2019fully, goranci2019fully}, for the case of Euclidean space, or metrics with bounded doubling dimension, sublinear in $k$-update time is in fact possible.

At a high level, the aforementioned improvements for specialized metric spaces can be distilled to the following observation: given a current set of $k$ centers $C \subset \cX$ that yield a solution with cost $R$, and a new data point $x \in \cX$, in order to maintain a $c$-approximate $k$-centers solution, one must immediately distinguish between the case that $d(x,C) < R$ and $d(x,C) > cR$. In general, when the metric is a priori unknown to the algorithm, one can do no better than a brute force search, checking for each $y \in C$ whether $d(x,y) < R$, resulting in an $\Omega(k)$ update time. On the other hand, for metrics which admit fast, sublinear (approximate) nearest neighbor search data structures, one can distinguish between these two cases in $o(k)$ time. This is precisely the fact that is exploited in \cite{schmidt2019fully, goranci2019fully}. However, these results leveraged specific nearest neighbor data structures, along with specialized clustering algorithms to employ them. Moreover, for the case of Euclidean space, the resulting approximation factors were still logarithmic. Therefore, it is natural to ask whether \textit{any} space that admits sublinear time nearest neighbor search data structures also admits sublinear-in-$k$ update time fully dynamic $k$-centers algorithms. 

In Section \ref{sec:LSH}, we demonstrate an affirmative answer to the above for nearly all such metric spaces --- namely, all metric spaces that admit \textit{locally sensitive hash functions} (LSH). Locally sensitive hashing is a very well studied technique \cite{datar2004locality, broder1997resemblance, indyk1998approximate,andoni2006near, har2012approximate,andoni2014beyond,dasgupta2011fast}, and is perhaps the most common method for designing approximate nearest neighbor search algorithms. Specifically, we develop a black box transformation 
which takes a LSH family for a metric space and produces a faster fully dynamic $k$-centers algorithm. Our main result for LSH spaces is stated below, and utilizes the standard notion of $(r,cr,p_1,p_2)$-sensitive hash families (see Definition \ref{def:LSH}). 

\noindent \textbf{Theorem} \ref{thm:lshMain}. {\it
Let $(\cX,d)$ be a metric space, and fix $\delta \in (0,1/2)$. Suppose that for any $r \in (r_{\min}, r_{\max})$ there exists an $(r,cr,p_1,p_2)$-sensitive hash family $\cH_r: \cX \to U$, such that each $h \in \cH_r$ can be evaluated in time at most $\Run(\cH)$, and such that $p_2$ is bounded away from $1$. Then there is a fully dynamic algorithm that, on a sequence of $M$ insertions and deletions of points from $\cX$, given an upper bound $M \leq \hat{M} \leq \poly(M)$, with probability $1-\delta$, correctly maintains a $c(2+\eps)$-approximate $k$-centers clustering to the active point set $P^t$ at all time steps $t$, simultaneously for all $k \geq 1$. The total runtime of the algorithm is at most 
\[\tilde{O}\left(M \cdot \frac{\log \Delta \log \delta^{-1}}{\eps p_1} n^{2 \rho} \cdot \Run(\cH) \right)\]
where $\rho = \frac{\ln p_1}{\ln p_2}$, and $n$ is an upper bound on the maximum number of points at any time step. 
}

We remark that the quantity $\rho$ is a primary quantity of concern in most works which study locally sensitive hashing. Therefore, for many well known metric spaces, an upper (and sometimes matching lower) bound on $\rho$ is known, immediately resulting in improved $k$-centers algorithms for these metrics.

As an application of Theorem \ref{thm:lshMain}, we obtain improved fully dynamic algorithms for Eucledian space, $\ell_p$ spaces, the Hamming metric, and the Jaccard metric, by using known locally sensitive hash functions for these spaces \cite{broder1997resemblance, datar2004locality, har2012approximate}. We remark that previously, there were no known fully dynamic algorithms for the Hamming metric or the Jaccard Metric that achieved sublinear in $k$ update time. For the case of Euclidean space, we improve on the prior best approximation by a factor of $\Omega(\log n)$. 

For all the aforementioned metrics, we give the first sublinear (in $n$) update time algorithms that obtain a constant approximation to $k$-centers simultaneously for all $k$.  In particular, we obtain an approximation of at most $c(4+\eps)$ in amortized update time at most $\tilde{O}(n^{1/c})$ (omitting $1/\eps,\log \Delta, \log M$ factors), for any $c \geq 1$. Setting $c = \eps \log n/ \log \log n$ for any constant $\eps > 0$, we obtain a $\eps \frac{\log n}{ \log \log n}$-approximation in amortized $\polylog(n,M,\Delta)$ update time. 
Our main results, along with the prior best known bounds, are summarized in Figure \ref{fig:results}.
We state the results formally in corollaries which follow.

\noindent \textbf{Corollary} \ref{cor:Euclidean} and \ref{cor:Jaccard}.  {\it
 Fix any $c \geq 1$. Then there is a fully dynamic algorithm that, on a sequence of $M$ insertions and deletions of points from either $d$-dimensional $\ell_p$ space $(\R^d , \ell_p)$ for $p \in [1,2]$, the Hamming metric over $\R^d$, or the Jaccard metric over subsets of a finite set $X$, with probability $1-\delta$, correctly maintains a $c(4+\eps)$-approximate $k$-centers clustering to the active point set $P^t$ at all time steps $t \in [M]$, and simultaneously for all $k \geq 1$. The total runtime is at most 
\[\tilde{O}\left(M \frac{ \log \delta^{-1} \log \Delta  }{\eps} \Run(\cH)  \cdot n^{1/c}\right),\]
where $\Run(\cH) = |X|$ for the case of the Jaccard metric, and $\Run(\cH)=d$ otherwise. 
}

In particular, our algorithm for $\ell_1$-space immediately results in a collection of fully dynamic $k$-centers algorithms for metric spaces which admit low-distortion embeddings into $\ell_1$, including the \textit{earth mover's distance} (EMD)  \cite{IT03,andoni2008earth}, and the \textit{edit distance} \cite{10.1145/1284320.1284322}. We omit formal definitions of these metrics, and refer the reader to the aforementioned references for the embeddings into $\ell_1$, after which our algorithms for $\ell_1$ can be applied with approximation blow-up given by the distortion of the embedding. 


For the case of standard Euclidean space ($p=2$), one can use the celebrated ball carving technique of Andoni and Indyk \cite{andoni2006near}, which results in improved locally sensitive hash functions for moderate values of $c$, resulting in the following:

\noindent \textbf{Corollary} \ref{cor:Euclidean2}. {\it 
Fix any $c \geq 1$. Then there is a fully dynamic algorithm which, on a sequence of $M$ insertions and deletions of points from $d$-dimensional Euclidean space $(\R^d , \ell_2)$,  with probability $1-\delta$, correctly maintains a $c(\sqrt{8}+\eps)$-approximate $k$-centers clustering to the active point set $P^t$ at all time steps $t \in [M]$, and simultaneously for all $k \geq 1$. The total runtime is at most 
\[\tilde{O}\left(M \frac{ \log \delta^{-1} \log \Delta  }{\eps} d n^{1/c^2 + o(1)}\right).\]
}

	\subsection{Other Related Work}
	
	The first dynamic algorithms for $k$-clustering tasks were given by Charikar, Chekuri, Feder, and Motwani~\cite{charikar2004incremental}. For $k$-centers, they gave a deterministic algorithm which maintains a $8$-approximation in amortized $O(k)$ time per update. The first fully-dynamic algorithm for $k$-centers clustering with non-trivial update time was a randomized algorithm given by Chan, Guerqin, and Sozio \cite{chan2018fully}, which yields a $(2+\eps)$-approximation in $O(\frac{\log\Delta}{\eps}k^2 )$ amortized runtime per update. 
	For specialized metric spaces,
	Schmidt and Sohler \cite{schmidt2019fully} demonstrated that in $d$-dimensional Euclidean space it is possible to obtain a $O(c \log n)$ approximation in $\tilde{O}(n^{1/c})$ expected update time (ignoring $d,\log \Delta$ factors), and \cite{goranci2019fully} gave a  $(2+\epsilon$)-approximation for any metric space with constant doubling dimension, running in polylogarithmic (in $n$) update time, but exponential in the doubling dimension.
	
	The problem of $k$-centers clustering has also been studied in the incremental streaming model, where the emphasis is on small space algorithms, but points are only inserted and never deleted. In \cite{charikar2003better,mccutchen2008streaming}, streaming algorithms for $k$-centers with outliers were given, where the clustering can ignore a small number of points. Streaming algorithms for a matroid generalization of $k$-centers were also considered in \cite{kale2019small}. Additionally, $k$-centers has also been studied in the \textit{sliding window} streaming model in \cite{cohen2016diameter}, where points are only inserted but expire after a fixed number of subsequent updates. 
	

	Fully dynamic clustering for $k$-means and $k$-medians objectives have also been considered. In particular, in \cite{henzinger2020fully} it was shown that \textit{coresets} for $k$-means and $k$-medians can be maintained in a fully dynamic stream with update time roughly $O(k^2)$. Streaming algorithms for $k$-means and $k$-medians via coresets in the insertion-only streaming setting are very well studied (see, e.g. \cite{har2004coresets,feldman2011unified,braverman2016new,feldman2020turning,huang2020coresets}). In general, comparable coresets for $k$-centers are not possible, although some guarantees for \textit{composable coresets} are possible for specialized metrics, such as those with bounded doubling dimension \cite{aghamolaei2019composable}.
	Beyond $k$-clustering tasks, in \cite{cohen2019fully}, fully dynamic algorithms were given for facility location, which is closely related to the $k$-medians objective. 

Lastly, in addition to dynamic algorithms, another line of work focuses on optimizing the \textit{consistency} of the clusterings produced during a fully dynamic or incremental stream \cite{lattanzi2017consistent,cohen2019fully, Guo2020consistent}, which is defined as the total number of changes made to the set of centers over the stream. 


\subsection{Technical Overview}\label{sec:tech}
We now describe the main techniques employed in our algorithm for general metric spaces, in our faster algorithms for LSH spaces, and in our query complexity lower bound for $k$-clustering in general metric spaces. For the remainder of the section, we fix a metric space $(\cX,d)$ and a point set $P \subset \cX$ of size $n$, with $r_{\min} \leq d(x,y) \leq r_{\max}$ for all $x,y \in P$, and set $\Delta = r_{\max}/r_{\min}$. 

\subsubsection{Algorithm for General Metric Spaces} Our starting point is the well-known reduction of Hochbaum and Shmoys \cite{hochbaum1986unified} from approximating $k$-centers to computing a maximal independent set (MIS) in a collection of \textit{threshold} graphs. Formally, given a real $r>0$, the $r$-threshold graph of a point set $P$ is the graph $G_r = (V,E_r)$ with vertex set $V = P$, and where  $(x,y) \in E_r$ is an edge if and only if $d(x,y) \leq r$. One computes an MIS $\cI_r$ in the graph $G_r$ for each $r = (1+\eps)^i r_{\min}$ with $i=0,1,\dots,\lceil \log_{1+\eps} \Delta \rceil$. If $|\cI_r| \leq k$, then $\cI_r$ is a $k$-centers solution of cost at most $r$. If $|\cI_r| > k+1$, then there are $k+1$ points which are pair-wise distance at least $r$ apart, therefore, by the triangle inequality, the optimal cost is at least $r/2$. These facts together yield a $2+\eps$ approximation.

By the above, it suffices to maintain an MIS in $O(\eps^{-1} \log \Delta)$ threshold graphs. Now the problem of maintaining an MIS in a fully dynamic sequence of \textit{edge} insertions and deletions to a graph is very well studied \cite{assadi2019fully,gupta2018simple, onak2018fully,du2018improved, censor2016optimal, chechik2019fully, behnezhad2019fully}. Notably, this line of work has culminated with the algorithms of \cite{chechik2019fully,behnezhad2019fully}, which maintain an MIS in expected $\polylog n$ update time per edge insertion or deletion. Unfortunately, point insertions and deletions from a metric space correspond to \textit{vertex} insertions and deletions in a threshold graph. Since a single vertex update can change up to $O(n)$ edges in the graph at once, one cannot simply apply the prior algorithms for fully dynamic edge updates. Moreover, notice that in this vertex-update model, we are only given access to the graph via queries to the adjacency matrix. Thus, even finding a single neighbor of $v$ can be expensive.

On the other hand, observe that in the above reduction to MIS, one does not always need to compute the entire MIS; for a given threshold graph $G_r$, the algorithm can stop as soon as it obtains an independent set of size at least $k+1$. This motivates the following problem, which is to return either
an MIS of size at most $k$, or an independent set of size at least $k+1$. We refer to this as the $k$-Bounded MIS problem. Notice that given an MIS $\cI$ of size at most $k$ in a graph $G$ and a new vertex $v$, if $v$ is not adjacent to any $u \in \cI$, then $\cI \cup \{v\}$ is an MIS, otherwise $\cI$ is still maximal. Thus, while an insertion of a vertex $v$ can add $\Omega(n)$ edges to $G$, for the $k$-Bounded MIS problem, one only needs to check the $O(k)$ potential edges between $v$ and $\cI$ to determine if $\cI$ is still maximal. Thus, our goal will be to design a fully dynamic algorithm for $k$-Bounded MIS with $\tilde{O}(k)$ amortized update time in the vertex-update model.

\paragraph{The Algorithm for $k$-Bounded MIS.}
To accomplish the above goal, we will adapt several of the technical tools employed by the algorithms for fully dynamic MIS in the edge-update model. Specifically, one of the main insights of this line of work is to maintain the \textit{Lexicographically First Maximum Independent Set} (LFMIS) with respect to a random permutation $\pi: V \to [0,1]$ of the vertices.\footnote{LFMIS with respects to random orderings were considered in \cite{censor2016optimal,assadi2019fully,chechik2019fully,behnezhad2019fully}.} The LFMIS is a natural object obtained by greedily adding the vertex with smallest $\pi(v)$ to the MIS, removing it and all its neighbors, and continuing iteratively until no vertices remain. Maintaining an LFMIS under a random ranking has several advantages from the perspective of dynamic algorithms. Firstly, it is \textit{history-independent}, namely, once $\pi$ is fixed, the current LFMIS depends only on the current graph, and not the order of insertions and deletions which led to that graph. Secondly, given a new vertex $v$, the probability that adding $v$ to the graph causes a large number of changes to be made to the LFMIS is small, since $\pi(v)$ must have been similarly small for this to occur.

Given the above advantages of an LFMIS, our goal will be to maintain the set $\LFMIS_{k+1}$ consisting of the first $\min\{k+1,|\LFMIS|\}$ vertices in the overall LFMIS with respect to a random ranking $\pi$; we refer to $\LFMIS_{k+1}$ as the top-$k$ LFMIS. Notice that maintaining this set is sufficient to solve the $k$-Bounded MIS problem. The challenge in maintaining the set $\LFMIS_{k+1}$ will be to handle the ``excess'' vertices which are contained in the $\LFMIS$ but are not in $\LFMIS_{k+1}$, so that their membership in $\LFMIS_{k+1}$ can later be quickly determined when vertices with smaller rank in $\LFMIS_{k+1}$ are removed.
To handle these excess vertices, we make use of a priority queue $\cQ$, with priorities given by the ranking $\pi$. 
When the LFMIS becomes larger than $k+1$, we store additional arriving vertices that do not belong to $\LFMIS_{k+1}$ in $\cQ$. Specifically, if $u_{k+1}$ is the vertex with the $(k+1)$-st largest rank in $\LFMIS_{k+1}$, then if $\pi(v) > \pi(u_{k+1})$ it follows that $v$ could not be part of $\LFMIS_{k+1}$, and therefore $v$ can be added to $\cQ$. Whenever a new vertex with rank less than $\pi(u_{k+1})$ is added to $\LFMIS_{k+1}$, forcing the vertex $u_{k+1}$ out of the top-$k$, we remove $u_{k+1}$ and add it to $\cQ$. 
Conversely, whenever a vertex is removed from $\LFMIS_{k+1}$, we repeatedly attempt to insert the vertex in $\cQ$ with smallest rank until either $\cQ$ is empty or until $|\LFMIS_{k+1}| =k+1$. 

Now in general, the key difficulty in dynamically maintaining a MIS is that when a vertex $v$ in a MIS is deleted, potentially all of the neighbors of $v$ may need to be added to the MIS, resulting in a large update time. Firstly, in order to keep track of which vertices could possibly enter the LFMIS when a vertex is removed from it, we maintain a mapping $\ell:V \to \LFMIS$, such that for each $u \notin \LFMIS$, we have $\ell(u) \in \LFMIS$ and $(u,\ell(u))$ is an edge. The ``leader'' $\ell(u)$ of $u$ serves as a certificate that $u$ cannot be added to the MIS. When a vertex $v \in \LFMIS$ is removed from the LFMIS, we only need to search through the set $\cF_v = \{u \in V \;| \; \ell(u) = v\}$ to see which vertices should be added to the LFMIS. Note that this can occur when $v$ is deleted, or when a neighbor of $v$ with smaller rank is added to the LFMIS. Consequentially, the update time of the algorithm is a function of the number points $u$ whose leader $\ell(u)$ changes on that step. For each such $u$, we can check in $O(k)$ time if it should be added to $\LFMIS_{k+1}$ by querying the edges between $u$ and the vertices in $\LFMIS_{k+1}$. By a careful amortized analysis, we can prove that the total runtime of this algorithm is indeed at most an $O(k)$ factor larger than the total number of leader changes. This leaves the primary challenge of designing and maintaining a leader mapping which changes infrequently. 


A natural choice for such a leader function is to set $\ell(u)$ to be the \textit{eliminator} of $v$ in the LFMIS. Here, for any vertex $u$ not in the LFMIS, the eliminator $\elim_\pi(u)$ of $u$ is defined to be its neighbor with lowest rank  that belongs to the LFMIS. 
The eliminators have the desirable property that they are also history-independent, and therefore the number of changes to the eliminators on a given update depends only on the current graph and the update being made. Maintaining the eliminators was an important component of the MIS algorithm of \cite{behnezhad2019fully}. In particular, a key result of \cite{behnezhad2019fully} is that the expected number of changes to the eliminators of the graph, even after the insertion or removal of an entire vertex, is at most $O(\log n)$. Therefore, if we could maintain the mapping $\ell(v) = \elim_{\pi}(v)$ by keeping track of the eliminators, our task would be complete.

Unfortunately, keeping track of the eliminators will not be possible in the vertex-update model, since we can only query a small fraction of the adjacency matrix after each update. In particular, when a vertex $v$ is inserted, it may change the eliminators of many of its neighbors, but we cannot afford to query all $\Omega(n)$ potential neighbors of $v$ to check which eliminators have changed. Instead, our solution is to maintain a leader mapping $\ell(v)$ which is an ``out-of-date'' version of the eliminator mapping. Each time we check if a vertex $v$ can be added to $\LFMIS_{k+1}$, by searching through its neighbors in $\LFMIS_{k+1}$, we ensure that either $v$ is added to $\LFMIS_{k+1}$ or its leader $\ell(v)$ is updated to the current eliminator of $v$, thereby aligning $\ell(v)$ with $\elim_\pi(v)$. However, thereafter, 
the values of $\ell(v)$ and $\elim_{\pi}(v)$ can become misaligned in several circumstances. In particular, the vertex $v$ may be moved into the queue $\cQ$ due to its leader $\ell(v)$ either leaving the LFMIS, or being pushed out of the top $k+1$ vertices in the LFMIS. In the second case, we show that $v$ can follow its leader to $\cQ$ without changing $\ell(v)$, however, in the first case $\ell(v)$ is necessarily modified. On the other hand, as noted, the eliminator of $v$ can also later change without the algorithm having to change $\ell(v)$.
Our analysis proceeds by a carefully accounting, in which we demonstrate that each change in an eliminator can result in at most a constant number of changes to the leaders $\ell$, from which an amortized bound of $O(\log n)$ leader changes follows via the results of \cite{behnezhad2019fully}. 



\paragraph{Comparison to the Prior Algorithm of \cite{chan2018fully}.}
The prior fully dynamic $k$-centers algorithm of Chan, Gourqin, and Sozio \cite{chan2018fully}, which obtained an amortized $O(\eps^{-1}\log \Delta \cdot k^2)$ update time, also partially employed the idea of maintaining an LFMIS (although the connection to MIS under lexicographical orderings was not made explicit in that work). However, instead of consistently maintaining the LFMIS with respect to a random ranking $\pi$, they begin by maintaining an LFMIS with respect to the ordering $\pi'$ in which the points were originally inserted into the stream. Since this ordering is adversarial, deletions in the stream can initially be very expensive to handle. To prevent bad deletions from repeatedly occurring, whenever a deletion to a center $c$ occurs, the algorithm of \cite{chan2018fully} randomly reorders all points which are contained in clusters that come after $c$ in the current ordering being used. 
In this way, the algorithm of \cite{chan2018fully} gradually converts the adversarial ordering $\pi'$ into a random ordering $\pi$. However, by reordering \textit{all} points which occurred after a deleted center $c$, instead of just the set of points which were led by that center (via a mapping $\ell$), the amortized update time of the algorithm becomes $O(k^2)$.\footnote{Consider the stream which inserts $k$ clusters of equal size $n/k$, and then begins randomly deleting half of each cluster in reverse order. By the time a constant fraction of all the points are deleted, for each deletion the probability a leader is deleted is $\Omega(k/n)$, but such a deletion causes $O(nk)$ work to be done by the algorithm. } 
In contrast, one of our key insights is to update the entire clustering to immediately reflect a random LFMIS ordering after each update. 





\subsubsection{Algorithm for LSH Spaces}
The extension of our algorithm to LSH spaces is based on the following observation: each time we attempt to add a vertex $v$ to $\LFMIS_{k+1}$, we can determine the fate of $v$ solely by finding the vertex $u \in \LFMIS_{k+1}$ in the neighborhood of $v$ of minimal rank (i.e., the eliminator of $v$, if it is contained in $\LFMIS_{k+1}$). If $\pi(u) < \pi(v)$, we simply set $\ell(v) = u$ and proceed. Otherwise, we must find all other neighbors $w$ of $v$ in $\LFMIS_{k+1}$, remove them from the LFMIS, and set $\ell(w) = u$. Finding the vertex $u$ can therefore be cast as an $r$-\textit{near neighbor search} problem: here, one wants to return any $u \in \LFMIS_{k+1}$ which is at distance at most $r$ from $u$, with the caveat that we need to return such vertices in order based on their ranking. 
Since, whenever $u$ enters the LFMIS, each point $w$ that we search through which leaves $\LFMIS_{k+1}$ had its leader change, if we can find each consecutive neighbor of $u$ in $\LFMIS_{k+1}$ in time $\alpha$, we could hope to bound the total runtime of the algorithm by an $O(\alpha)$ factor more than the total number of leader changes, which we know to be small by analysis of the general metric space algorithm. 

To achieve values of $\alpha$ which are sublinear in $k$, we must necessarily settle for an \textit{approximate near neighbor search} (ANN) algorithm. A randomized, approximate $(r,cr)$-nearest neighbor data structure will return any point in $\LFMIS_{k+1}$ which is at distance at most $cr$, assuming there is at least one point at distance at most $r$ in $\LFMIS_{k+1}$. In other words, such an algorithm can be used to find all edges in $G_r$, with the addition of any arbitrary subset of edges in $G_{cr}$. By relaxing the notion of a threshold graph to allow for such a $c$-approximation, one can hope to obtain a $c(2+\eps)$-approximation to $k$-centers via solving the $k$-Bounded MIS problem on each relaxed threshold graph. 

However, there are several serious challenges in the above black-box reduction to an arbitrary ANN data structure. Firstly, the above algorithm \textit{adaptively} queries the ANN data structure: the points which are inserted into $\LFMIS_{k+1}$, as well as the future edges which are reported by the data structure, depend on the prior edges which were returned by the data structure. Such adaptive reuse breaks down traditional guarantees of randomized algorithms, hence designing such algorithms which are robust to adaptivity is the subject of a growing body of research \cite{ben2020framework,cherapanamjeri2020adaptive,HassidimKMMS20,WoodruffZ20,ACSS21}. More nefariously, the adaptivity also goes in the other direction: namely, the random ordering $\pi$ influences which points will be added to the set $\LFMIS_{k+1}$, in turn influencing the future queries made to the ANN data structure, which in turn dictate the edges which exist in the graph (by means of queries to the ANN oracle). Thus, the graph itself cannot be assumed to be independent of $\pi$! 

The key issue above is that, when using an ANN data structure, the underlying relaxed threshold graph is no longer a deterministic function of the point set $P$, and is instead ``revealed'' as queries are made to the ANN data structure. We handle this issue by demonstrating that, for the class of ANN algorithms based on locally sensitive hash functions, one can define a graph $G$ which is only a function of the randomness in the ANN data structure, and not the ordering $\pi$. The edges of this graph are defined in a natural way --- two points are adjacent if they collide in at least one of the hash buckets. By an appropriate setting of parameters, the number of collisions between points at distance larger than $cr$ can be made small. By simply ignoring such erroneous edges as they are queried, the runtime increases by a factor of the number of such collisions. Lastly, by storing the points in each hash buckets via a binary search tree, we can ensure that edge queries can be answered in the order of their ranking, satisfying the properties we needed for usage in our main algorithm for $k$-Bounded MIS.

\subsubsection{Query Complexity Lower Bound for General Metric Spaces}
The lower bound of $\Omega(k^2)$ is fairly straightforward, and follows from hiding a small distance of $\eps$ in a distance matrix $\bD$ consisting otherwise of all $1$'s on the off-diagonals. In this case, adaptive algorithms have no advantage over non-adaptive algorithms, the latter of which are simpler to analyze. Proving the $\Omega(nk)$ query lower bound requires more care, due to the handling of adaptive testers. The hard distribution is as follows: in one case we randomly plant $k$ clusters each of size roughly $n/k$, where points within a cluster are close and points in separate clusters are far. In the second case, we do the same, and subsequently choose a point $i \sim [n]$ randomly and move it very far from all points (including its own cluster). Adaptive algorithms can gradually winnow the set of possible locations for $i$ by discovering connected components in the clusters, and eliminating the points in those components. Our proof follows by demonstrating that a sizable fraction of the probability space of the input distribution lies in computation paths of an adaptive algorithms which eliminate few data points, and therefore have small advantage in discovering the planted point $i$.

\section{Preliminaries}\label{sec:prelims}


We begin with basic notation and definitions. For any positive integer $n$, we write $[n]$ to denote the set $ \{1,2,\dots,n\}$. In what follows, we will fix any metric space $(\cX,d)$. 
A fully dynamic stream is a sequence $(p_1,\sigma_1),\dots,(p_M,\sigma_M)$ of $M$ updates such that $p_i \in \cX$ is a point, and $\sigma_i \in \{+,-\}$ signifies either an insertion or deletion of a point. Naturally, we assume that a point can only be deleted if it was previously inserted. Moreover, we may assume without loss of generality that each point is inserted at most once before being deleted, as duplicate points will not change the $k$-centers cost. 

We call a point $p \in \cX$ active at time $t$ if $p$ was inserted at some time $t' < t$, and not deleted anytime between $t'$ and $t$. We write $P^{t} \subset \cX$ to denote the set of active points at time $t$. We let $r_{\min},r_{\max}$ be reals such that for all $t \in [M]$ and $x,y \in P^t$, we have $r_{\min} \leq d(x,y) \leq r_{\max}$, and set $\Delta = r_{\max} / r_{\min}$ to be the aspect ratio of the point set. As in prior works \cite{chan2018fully,schmidt2019fully}, we assume that an upper bound on $\Delta$ is known.

\paragraph{Clustering Objectives.}
In the $k$-centers problem, given a point set $P$ living in a metric space $(\cX,d)$, the goal is to output a set of $k$ \textit{centers} $C = \{c_1,\dots,c_k\} \subset \cX$, along with a mapping $\ell: P \to \{c_1,\dots,c_k\}$, 
such that the following objective is minimized:
\[ \cost_{k,\infty}(\cC) =   \max_{p \in P} d(p,\ell(p))    \]
In other words, one would like for the maximum distance from $p$ to the $\ell(p)$, over all $p \in P$, to be minimized. Additionally, for any real $z >0$, we introduce the $(k,z)$-clustering problem, which is to minimize 
\[ \cost_{k,z}(\cC) =   \sum_{p \in P} d^z(p,\ell(p))    \]
We will be primarily concerned the the $k$-centers objective, but we introduce the more general $(k,z)$-clustering objective, which includes both $k$-medians (for $z=1$) and $k$-means (for $z=2$), as our lower bounds from Section \ref{sec:LB} will hold for these objectives as well. We remark that while $\ell(p)$ is usually fixed by definition to be the closest point to $p$ in $C$, the closest point may not necessarily be easy to maintain in a fully dynamic setting. Therefore, we will evaluate the cost of our algorithms with respects to both the centers and the mapping $\ell$ from points to centers. For any $p \in P$, we will refer to $\ell(p)$ as the \textit{leader} of $p$ under the mapping $\ell$, and the set of all points lead by a given $c_i$ is the cluster led by $c_i$.


In addition to maintaining a clustering with approximately optimal cost, we would like for our algorithms to be able to quickly answer queries related to cluster membership, and enumeration over all points in a cluster. Specifically, we ask that our algorithm be able to answer the following queries at any time step $t$: 

\begin{figure}[H]
\begin{Frame}[Queries to a Fully Dynamic Clustering Algorithm]
\begin{enumerate}
    \item \textbf{Membership Query:} Given a point $p \in P^{t}$, return the center $c = \ell(p)$ of the cluster $C$ containing~$p$.
    \item  \textbf{Cluster Enumeration:} Given a point $p \in P^{t}$, list all points in the cluster $C$ containing~$p$.
\end{enumerate}
\end{Frame}
\end{figure}

In particular, after processing any given update, our algorithms will be capable of repsonding to membership queries in $O(1)$-time, and to clustering enumeration queries in time $O(|C|)$, where $C$ is the clustering containing the query point.

\section{From Fully Dynamic $k$-Centers to $k$-Bounded MIS} 
\label{sec:kCenters}


In this section, we describe our main results for fully dynamic $k$-centers clustering, based on our main algorithmic contribution, which is presented in Section \ref{sec:generalMetric}. We begin by describing how
the problem of $k$-centers clustering of $P$ can be reduced to maintaining a maximal independent set (MIS) in a graph. In particular, the reduction will only require us to solve a weaker version of MIS, where we need only return a MIS of size at most $k$, or an independent set of size at least $k+1$. Formally, this problem, which we refer to as the $k$-Bounded MIS problem, is defined as follows.

\begin{definition}[$k$-Bounded MIS]
Given a graph $G = (V,E)$ and an integer $k \geq 1$, the $k$-bounded MIS problem is to output a maximal independent set $\cI \subset V$ of size at most $k$, or return an independent set $\cI \subset V$ of size at least $k+1$.
\end{definition}


\paragraph{Reduction from $k$-centers to $k$-Bounded MIS.}
The reduction from $k$-centers to computing a maximum independent set in a graph is well-known, and can be attributed to the work of Hochbaum and Shmoys \cite{hochbaum1986unified}. The reduction was described in Section \ref{sec:tech}, however, both for completeness and so that it is clear that only a $k$-Bounded MIS is required for the reduction, we spell out the full details here.

Fix a set of points $X$ in a metric space, such that $r_{\min} \leq d(x,y) \leq  r_{\max}$ for all $x,y \in X$. Then, for each $r= r_{\min}, (1+\eps/2) r_{\min} , (1+\eps/2)^2 r_{\min}, \dots, r_{\max}$, one creates the \textit{$r$-threshold graph} $G_r = (V,E_r)$, which is defined as the graph with vertex set $V = X$, and $(x,y) \in E_r$ if and only if $d(x,y) \leq r$. One then runs an algorithm for $k$-Bounded MIS on each graph $G_r$, and finds the smallest value of $r$ such that the output of the algorithm $\cI_r$ on $G_r$ satisfies $|\cI_r| \leq k$ --- in other words, $\cI_r$ must be a MIS of size at most $k$ in $G_r$. Observe that $\cI_r$ yields a solution to the $k$-centers problem with cost at most $r$, since each point in $X$ is either in $\cI_r$ or is at distance at most $r$ from a point in $\cI_r$. Furthermore, since the independent set $\cI_{r/(1+\eps/2)}$ returned from the algorithm run on $G_{r/(1+\eps/2)}$ satisfies $|\cI_{r/(1+\eps/2)}| \geq k+1$ it follows that there are $k+1$ points in $X$ which are pair-wise distance at least $r/(1+\eps/2)$ apart. Hence, the cost of any $k$-centers solution (which must cluster two of these $k+1$ points together) is at least $r/(2+\eps)$ by the triangle inequality. It follows that $\cI$ yields a $2+\eps$ approximation of the optimal $k$-centers cost.

Note that, in addition to maintaining the centers $\cI$, for the purposes of answering membership queries, one would also like to be able to return in $O(1)$ time, given any $x \in X$, a fixed $y \in \cI$ such that $d(x,y) \leq r$. We will ensure that our algorithms, whenever they return a MIS $\cI$ with size at most $k$, also maintain a mapping $\ell: V \setminus \cI \to \cI$ which maps any $x$, which is not a center, to its corresponding center $\ell(x)$. 

Observe that in the context of $k$-clustering, insertions and deletions of points correspond to insertions and deletions of entire vertices into the graph $G$. This is known as the fully dynamic \textit{vertex update model}. Since one vertex update can cause as many as $O(n)$ edge updates, we will not be able to read all of the edges inserted into the stream. Instead, we assume our dynamic graph algorithms can query for whether $(u,v)$ is an edge in constant time (i.e., constant time oracle access to the adjacency matrix).\footnote{This is equivalent to assuming that distances in the metric space can be computed in constant time, however if such distances require $\alpha$ time to compute, this will only increase the runtime of our algorithms by a factor of $\alpha$.}

Our algorithm for $k$-Bounded MIS will return a very particular type of MIS. Specifically, we will attempt to return the first $k+1$ vertices in a \textit{Lexicographically First MIS} (LFMIS), under a random lexicographical ordering of the vertices.

\paragraph{Lexicographically First MIS (LFMIS).} The LFMIS of a graph $G = (V,E)$ according to a ranking of the vertices specified by a mapping $\pi:V \to [0,1]$ is a unique MIS defined as by the following process. Initially, every vertex $v \in V$ is alive. We then iteratively select the alive vertex with minimal rank $\pi(v)$, add it to the MIS, and then kill $v$ and all of its alive neighbors. We write $\LFMIS^\pi(G)$ to denote the LFMIS of $G$ under $\pi$. For each vertex $v$, we define the \textit{eliminator} of $v$, denoted $\elim_{G,\pi}(v)$ to be the vertex $u$ which kills $v$ in the above process; namely,  $\elim_{G,\pi}(v)$ is the vertex with smallest rank in the set $(N(v) \cup \{v\}) \cap \LFMIS^\pi(G,\pi)$.

\begin{definition}
Given a graph $G = (V,E)$, $\pi:V \to [0,1]$, and an integer $k \geq 1$, we define the top-$k$ LFMIS of $G$ with respect to $\pi$, denoted $\LFMIS_k(G,\pi)$ to be the set consisting of the first $\min\{k,|\LFMIS(G,\pi)|\}$ vertices in $\LFMIS(G,\pi)$ (where the ordering is with respect to $\pi$). When $G,\pi$ are given by context, we simply write $\LFMIS_k$. 
\end{definition}

It is clear that returning a top-$(k+1)$ LFMIS of $G$ with respect to any ordering will solve the $k$-Bounded MIS problem. In order to also obtain a mapping $\ell$ from points to their centers in the independent set, we define the following augmented version ofthe top-$k$ LFMIS problem, which we refer to as a \textit{top-$k$ LFMIS with leaders}.

\begin{definition}\label{def:LFMISLead}
A top-$k$ LFMIS with leaders consists of the set $\LFMIS_k(G,\pi)$, along with a \textit{leader mapping function} $\ell:V \to V \cup \{\bot\}$, such that $(v,\ell(v)) \in E$ whenever $\ell(v) \neq \bot$, and such that if $\LFMIS_k(G,\pi) = \LFMIS(G,\pi)$, then $\ell(v) \in \LFMIS_k(G,\pi)$ for all $v \in V \setminus \LFMIS_k(G,\pi)$, and $\ell(v) = \bot$ for all $v \in \LFMIS_k(G,\pi)$. 
\end{definition}

The main goal of the following Section \ref{sec:generalMetric} will be to prove the existence of a $\tilde{O}(k)$ amortized update time algorithm for maintaining a top-$k$ LFMIS with leaders of a graph $G$ under a fully dynamic sequence of insertions and deletions of vertices from $G$. Specifically, we will prove the following theorem.

\noindent \textbf{Theorem} \ref{thm:LFMISMain}. {\it There is a algorithm which, on a fully dynamic stream of insertions and deletions of vertices to a graph $G$, maintains at all time steps a top-$k$ LFMIS of $G$ with leaders under a random ranking $\pi: V \to [0,1]$. The expected amortized per-update time of the algorithm is $O(k \log n + \log^2 n)$, where $n$ is the maximum number active of vertices at any time. Moreover, the algorithm does not need to know $n$ in advance. }

Next, we demonstrate how Theorem \ref{thm:LFMISMain} immediately implies the main result of this work (Theorem \ref{thm:main}). Firstly, we prove a proposition which demonstrates that any fully dynamic Las Vegas algorithm with small runtime in expectation can be converted into a fully dynamic algorithm with small runtime with high probability.

\begin{proposition}\label{prop:highProb}
Let $\cA$ be any fully dynamic randomized algorithm that correctly maintains a solution for a problem $\cP$ at all time steps, and runs in amortized time at most $\Run(\cA)$ in expectation. Then there is a fully dynamic algorithm for $\cP$ which runs in amortized time at most $O(\Run(\cA) \log\delta^{-1})$ with probability $1-\delta$ for all $\delta \in (0,1/2)$.
\end{proposition}
\begin{proof}
The algorithm is as follows: we maintain at all time steps a single instance of $\cA$ running on the dynamic stream. If, whenever the current time step is $t$, the total runtime of the algorithm exceeds $4 t \Run(\cA)$, we delete $\cA$ and re-instantiate it with fresh randomness. We then run the re-instantiated version of $\cA$ from the beginning of the stream until either we reach the current time step $t$, or the total runtime again exceeds $4 t \Run(\cA)$, in which case we re-instantiate again. 

Let $M$ be the total length of the stream. For each $i=0,1,2,\dots,\lceil \log M \rceil$, let $\bZ_i$ be the number of times that $\cA$ is re-instantiated while the current time step is between $2^i$ and $2^{i+1}$. Note that the total runtime is then at most 
\[ 4 \Run(\cA) \cdot \left(M +  \sum_{i=0}^{\lceil \log M \rceil} 2^{i+1} \bZ_i \right) \]
Fix any $i \in \{0,1,\dots,\lceil \log M \rceil\}$, and let us bound the value $\bZ_i$. 
Each time that $\cA$ is restarted when the current time $t$ step is between $2^i$ and $2^{i+1}$, the probability that the new algorithm runs in time more than $2^{i+2}\Run(\cA) \geq 4  t \Run(\cA) $ on the first $2^{i+1}$ updates is at most $1/2$ by Markov's inequality. Thu, the probability that $\bZ_i > T_i + 1$ is at most $2^{-T_i}$, for any $T_i \geq 0$. Setting $T_i = \log(2/\delta) + \lceil \log M \rceil - i $, we have
\begin{equation}
    \begin{split}
        \sum_{i=0}^{\lceil \log M \rceil} \pr{\bZ_i > T_i + 1 }& \leq \frac{\delta}{2}   \sum_{i=0}^{\lceil \log M \rceil} \left(\frac{1}{2}\right)^{\lceil \log M \rceil - i} \\ 
        & < \delta
    \end{split}
\end{equation}
In other words, by a union bound, we have $\bZ_i \leq T_i+1$ for all $i=0,1,2,\dots,\lceil \log M \rceil$, with probability at least $1-\delta$. Conditioned on this, the total runtime is at most 

\begin{equation}
    \begin{split}
         4 \Run(\cA) \cdot \left(M +  \sum_{i=0}^{\lceil \log M \rceil} 2^{i+1} (T_i+1) \right) &= O\left(\Run(\cA)\left( \sum_{i=0}^{\lceil \log M \rceil} 2^{i} (  \log \delta^{-1} + \lceil \log M \rceil - i )  \right)\right) \\
         & = O\left(\Run(\cA)\left( M \log \delta^{-1} + M \sum_{i=0}^{\lceil \log M \rceil} \frac{i }{2^{i}}  \right)\right) \\
          & = O\left(\Run(\cA)\cdot  M \log \delta^{-1} \right) \\
    \end{split}
\end{equation}
which is the desired total runtime. 

\end{proof}

Given Theorem \ref{thm:LFMISMain}, along with the reduction to top-$k$ LFMIS from $k$-centers described in this section, which runs $O(\eps^{-1} \log \Delta)$ copies of a top-$k$ LFMIS algorithm, we immediately obtain a fully dynamic $k$-centers algorithm with $\tilde{O}(k)$ expected amortized update time. By then applying Proposition \ref{prop:highProb}, we obtain our main theorem, stated below. 

	\noindent \textbf{Theorem} \ref{thm:main}. {\it 
    There is a fully dynamic algorithm which, on a sequence of insertions and deletions of points from a metric space $\cX$, maintains a $(2+\eps)$-approximation to the optimal $k$-centers clustering. The amortized update time of the algorithm is $O(\frac{\log \Delta \log n}{\eps}(k  + \log n))$ in expectation, and $O(\frac{\log \Delta \log n}{\eps}(k  + \log n) \log \delta^{-1})$  with probability $1-\delta$ for any $\delta \in (0,\frac{1}{2})$, where $n$ is the maximum number of active points at any time step. 
    
    The algorithm can answer membership queries in $O(1)$-time, and enumerate over a cluster $C$ in time $O(|C_i|)$. 
}

%
\section{Fully Dynamic $k$-Bounded MIS with Vertex Updates}
\label{sec:generalMetric}

Given the reduction from Section \ref{sec:kCenters}, the goal of this section will be to design an algorithm which maintains a top-$k$ LFMIS with leaders  (Definition \ref{def:LFMISLead}) under a graph which receives a fully dynamic sequence of \textit{vertex} insertions and deletions. As noted, maintaining a top-$(k+1)$ LFMIS immediately results in a solution to the $k$-Bounded MIS problem. We begin by formalizing the model of vertex-valued updates to dynamic graphs. 


\paragraph{Fully Dynamic Graphs with Vertex Updates}
 In the vertex-update fully dynamic setting, at each time step a vertex $v$ is either inserted into the current graph $G$, or deleted from $G$, along with all edges incident to $v$. This defines a sequence of graphs $G^1, G^2,\dots,G^M$, where $G^t= (V^t,E^t)$ is the state of the graph after the $t$-th update. Equivalently, we can think of there being an ``underlying'' graph $G = (V,E)$, where at the beginning all vertices are \textit{inactive}. At each time step, either an active vertex is made inactive, or vice-versa, and $G^{t}$ is defined as the subgraph induced by the active vertices at time $t$. The latter is the interpretation which will be used for this section. 


Since the degree of $v$ may be as large as the number of active vertices in $G$, our algorithm will be unable to read all of the edges incident to $v$ when it arrives. Instead, we require only query access to the adjacency matrix of the underlying graph $G$. Namely, we assume that we can test in constant time whether $(u,v) \in E$ for any two vertices $u,v$.

For the purpose of $k$-centers clustering, we will need to maintain a top-$k$ LFMIS with leaders $(\LFMIS_k(G,\pi),\ell)$, along with a Boolean value indicating whether 
$\LFMIS_k(G,\pi) = \LFMIS(G,\pi)$. To do this, we can instead attempt to maintain the set $\LFMIS_{k+1}(G,\pi)$, as well as a leader mapping function $\ell:V \to V \cup \{\bot\}$, with the relaxed property that if $\LFMIS_k(G,\pi) = \LFMIS(G,\pi)$, then $\ell(v) \in \LFMIS(G,\pi)$ for all $v \in V \setminus \LFMIS(G,\pi)$ and $\ell(v) = \bot$ for all $v \in \LFMIS(G,\pi)$. We call such a leader function $\ell$ with this relaxed property a \textit{modified} leader mapping. Thus, in what follows, we will focus on maintaining a top-$k$ LFMIS with this modified leader mapping. 


\subsection{The Data Structure}
We now describe the main data structure and algorithm which will maintain a top-$k$ LFMIS with leaders in the dynamic graph $G$.
We begin by fixing a random mapping $\pi:V \to [0,1]$, which we will use as the ranking for our lexicographical ordering over the vertices. It is easy to see that if $|V| = n$, then by discretizing $[0,1]$ so that $\pi(v)$ can be represented in $O(\log n)$ bits we will avoid collisions with high probability. At every time step, the algorithm will maintain an ordered set $\ALG $ of vertices in a linked list, sorted by the ranking $\pi$, with $|\ALG| \leq k+1$. We will prove that, after every update $t$, we have $\ALG = \LFMIS_{k+1}(G^t,\pi)$.\footnote{We use a separate notation $\ALG$, instead of $\LFMIS_{k+1}$, to represent the set maintained by the algorithm, until we have demonstrated that we indeed have $\ALG = \LFMIS_{k+1}(G^t,\pi)$ at all time steps $t$.}  We will also maintain a mapping $\ell:V \to V \cup \{\bot\}$ which will be our leader mapping function. Initially, we set $\ell(v) = \bot$ for all $v$. 
Lastly, we will maintain a (potentially empty) priority queue $\cQ$ of \textit{unclustered} vertices, where the priority is similarly given by $\pi$.

Each vertex $v$ in $G_t$ will be classified as either a \textit{leader}, a \textit{follower}, or \textit{unclustered}. Intuitively, when $\LFMIS_k(G,\pi) = \LFMIS(G,\pi)$, the leaders will be exactly the points in $\LFMIS_k(G,\pi)$, the followers will be all other points $v$ which are mapped to some $\ell(v) \in \LFMIS_k(G,\pi)$ (in other words, $v$ ``follows'' $\ell(v)$), and there will be no unclustered points. At intermediate steps, however, when $|\LFMIS(G,\pi)| \geq k + 1$, we will be unable to maintain the entire set $\LFMIS(G,\pi)$ and, therefore, we will store the set of all vertices which are not in $\LFMIS_{k+1}$ in the priority queue $\cQ$ of unclustered vertices. The formal definitions of leaders, followers, and unclustered points follow.

Every vertex currently maintained in $\ALG$ is be a leader. Each leader $v$ may have a set of follower vertices, which are vertices $u$ with $\ell(u) = v$, in which case we say that $u$ follows $v$. By construction of the leader function $\ell$, every follower-leader pair $(u,\ell(u))$ will be an edge of $G$.  We write $\cF_v = \{u \in V : \ell(u) = v\}$ to denote the (possibly empty) set of followers of a leader $v$. For each leader, the set $\cF_v$ will be maintained as part of the data structure at the vertex $v$.

Now when the size of $\LFMIS$ exceeds $k+1$, we will have to remove the leader $v$ in $\LFMIS$ with the largest rank, so as to keep the size of $\ALG$ at most $k+1$. The vertex $v$ will then be moved to the queue $\cQ$, along with its priority $\pi(v)$. The set $\cF_v$ of followers of $v$ will continue to be followers of $v$ --- their status remains unchanged. In this case, the vertex $v$ is now said to be an \textit{inactive} leader, whereas each leader currently in $\ALG$ is called an \textit{active} leader. If, at a later time, we have $\pi(v) < \max_{u \in \ALG} \pi(u)$, then it is possible that $v$ may be part of $\LFMIS_{k+1}$, in which case we will attempt to reinsert the inactive leader $v$ from $\cQ$ back into $\ALG$. Note, importantly, that whenever $\pi(v) < \max_{u \in \ALG} \pi(u)$ occurs at a future time step for a vertex $v \in \cQ$, then either $v$ is part of $\LFMIS_{k+1}$, or it is a neighbor of some vertex $u \in \LFMIS_{k+1}$ of lower rank. In both cases, we can remove $v$ from $\cQ$ and attempt to reinsert it, with the guarantee that after this reinsertion $v$ will either be an active leader, or a follower of an active leader.


\paragraph{The Unclustered Queue.} We now describe the purpose and function of the priority queue $\cQ$. 
Whenever either a vertex $v$ is inserted into the stream, or it is a follower of a leader $\ell(v)$ who is removed from the $\LFMIS$, we must attempt to reinsert $v$, to see if it should be added to $\LFMIS_{k+1}$. However, if $|\ALG| = k+1$, then the only way that $v$ should be a part of $\LFMIS_{k+1}$ (and therefore added to $\ALG$) is if $\pi(v) < \max_{u \in \ALG} \pi(u)$. If this does not occur, then we do not need to insert $v$ right away, and instead can defer it to a later time when either $|\ALG| < k+1$ or $\pi(v) < \max_{u \in \ALG} \pi(u)$ holds. Moreover, by definition of the modified leader mapping $\ell$, we only need to set $\ell(v)$ when $|\LFMIS_{k+1}| < k+1$. We can therefore add $v$ to the priority queue $\cQ$. 

Every point in the priority queue is called an \textit{unclustered point}, as they are not currently part of a valid $k$-clustering in the graph. By checking the top of the priority queue at the end of processing each update, we can determine whenever either of the events $|\ALG| \leq k$ or $\pi(v) < \max_{u \in \ALG} \pi(u)$ holds; if either is true, we iteratively attempt to reinsert the top of the queue until the queue is empty or both events no longer hold. This will ensure that either all points are clustered (so $\cQ = \emptyset$), or $|\ALG| = k+1$ and $\ALG = \LFMIS_{k+1}$ (since no point in the queue could have been a part of $\LFMIS_{k+1}$).

\paragraph{The Leader Mapping.}  
Notice that given a top-$k$ LFMIS, a valid leader assignment is always given by $\ell(v) = \elim_{G,\pi}(v)$, where $\elim_{G,\pi}(v)$ is the eliminator of $v$ via $\pi$ as defined in Section \ref{sec:kCenters}; this is the case since if  $\LFMIS_k(G,\pi) = \LFMIS(G,\pi)$ then each vertex is either in $\LFMIS_k(G,\pi)$ or eliminated by one of the vertices in $\LFMIS_k(G,\pi)$.
Thus, intuitively, our goal should be to attempt to maintain that $\ell(v) = \elim_{G,\pi}(v)$ for all $v \notin \LFMIS_k(G,\pi)$.
However, the addition of a new vertex which enters $\LFMIS_{k}(G,\pi)$ can change the eliminators of many other vertices \textit{not} in $\LFMIS_k(G,\pi)$. Discovering which points have had their eliminator changed immediately on this time step would be expensive, as one would have to search through the followers of all active leaders to see if any of their eliminators changed. Instead, we note that at this moment, so long as the new vertex does not share an edge with any other active leader, we do not need to modify our leader mapping. Instead, we can \textit{defer} the reassignment of the leaders of vertices $v$ whose eliminator changed on this step, to a later step when their leaders are removed from $\LFMIS_{k+1}(G,\pi)$. Demonstrating that the number of changes to the leader mapping function $\ell$, defined in this way, is not too much larger than the number of changes to the eliminators of all vertices, will be a major component of our analysis.

\paragraph{The Algorithm and Roadmap.} Our main algorithm is described in three routines: Algorithms \ref{alg:main}, \ref{alg:ins}, and \ref{alg:del}. Algorithm \ref{alg:main} handles the inital insertion or deletion of a vertex in the fully dynamic stream, and then calls at least one of Algorithms \ref{alg:ins} or \ref{alg:del}. Algorithm \ref{alg:ins} handles insertions of vertices in the data structure, and Algorithm \ref{alg:del} handles deletions of vertices from the data structure. We begin in Section \ref{sec:correctness} by proving that our algorithm does indeed solve the top-$k$ LFMIS problem with the desired modified leader mapping. Then, in Section \ref{sec:amortized}, we analyze the amortized runtime of the algorithm.




\begin{algorithm}[ht]
\DontPrintSemicolon
	\caption{Process Update}\label{alg:main}
	\KwData{An update $(v,\sigma)$, where $\sigma \in \{+,-\}$.}
 \If{$\sigma = +$ is an insertion of $v$}{
 Generate $\pi(v)$, and set $\ell(v) = \bot$.\;
 Call $\ins(v,\pi(v))$.\;}
\If{$\sigma = -$ is a deletion of $v$}{
Call $\delete(v)$.\;
}

\While{$|\cQ| \neq \emptyset \boldsymbol{\wedge} \left( |\ALG| \leq  k \boldsymbol{\vee} \min_{w \in \cQ} \pi(w) < \max_{w \in \ALG} \pi(w)  \right)$}{ \label{line:whileMain}  
$u \leftarrow \arg \min_{w \in \cQ} \pi(w)$.\;
Delete $u$ from $\cQ$, and call $\ins(u,\pi(u))$. \;
}
\end{algorithm}


\begin{algorithm}[ht]
\DontPrintSemicolon
	\caption{$\ins(v,\pi(v))$} \label{alg:ins}
 \If{$|\ALG| = k+1$ $\boldsymbol{\wedge}$ $\pi(v) > \max_{u \in \ALG} \pi(u)$}{\label{line:firstIfIns}
 Insert $(v,\pi(v))$ into $\cQ$. \;}
\Else{
Compute $S =  \ALG \cap N(v)$\;
\If{$S = \emptyset$}{
Add $v$ to $\ALG$. \;
\If{ $|\ALG| = k+2$}{
Let $u = \arg \max_{u' \in \ALG} \pi(u') $. \;
Remove $u$ from $\ALG$, and insert $(u,\pi(u))$ into $\cQ$. \label{line:delOverflow} \;
}
}
\Else{
    $u^* = \arg \min_{u' \in S} \pi(u')$\;
    \If{$\pi(u^*) < \pi(v)$}{
    \If{$v$ is a leader}{
    For each $w \in \cF_v$, insert $(w,\pi(w))$ to $\cQ$, and set $\ell(w) = \bot$ \label{line:queue1Ins} \;
    Delete the list $\cF_v$. \;
    }
    Add $v$ to $\cF_{u^*}$ as a follower of $u^*$, set $\ell(v) = u^*$. \;
    }
    \Else{
    For each $w \in \cup_{u \in S} \cF_u$, add $(w,\pi(w))$ to $\cQ$, and set $\ell(w) = \bot$. \label{line:queue2Ins} \;
    For each $u \in S$, set $\ell(u) = v$ to be a follower of $v$, remove $u$ from $\ALG$, and delete the list $\cF_u$. \;
    }
 }  
}
\end{algorithm}

\begin{algorithm}[!ht]
\DontPrintSemicolon
	\caption{$\delete(v)$ }\label{alg:del}
	
	\If{$v$ is a follower}{
	Delete $v$ from $\cF_{\ell(v)}$, and remove $v$ from the set of vertices.  \; 
	}
	\ElseIf{$v \in \cQ$}{
	\If{$v$ is a leader}{
	 For each $w \in \cF_v$, insert $(w,\pi(w))$ to $\cQ$, and set $\ell(w) = \bot$ \label{line:queue1Del} \;
	 Delete the list $\cF_v$, and remove $v$ from $\cQ$ and the set of vertices. \;
	}
	\Else{
	Delete $v$ from $\cQ$ and the set of vertices. \;
	}
	}
	\Else(\tcc*[r]{Must have $v \in \ALG$}){
	 For each $w \in \cF_v$, insert $(w,\pi(w))$ to $\cQ$, and set $\ell(w) = \bot$\label{line:queue2Del} \;
	 	 Delete the list $\cF_v$, and remove $v$ from $\ALG$ and the set of vertices. \;
	}
\end{algorithm}


\subsection{Correctness of the Algorithm}\label{sec:correctness}
We will now demonstrate the correctness of the algorithm, by first proving two Propositions. 

\begin{proposition}\label{prop:IS}
After every time step $t$, the set $\ALG$ stored by the algorithm is an independent set in $G^{(t)}$.
\end{proposition}
\begin{proof}
Suppose otherwise, and let $v,u \in \ALG$ be vertices with $(v,u) \in E$. WLOG we have that $v$ is the vertex which entered $\ALG$ most recently of the two. Then we had $ u \in ALG$ at the moment that $\ins(v,\pi(v))$ was most recently called. Since on the step that $\ins(v,\pi(v))$ was most recently called the vertex $p$ was added to $\ALG$, it must have been that $\min_{w \in N(v) \cap \ALG} \pi(w) > \pi(v)$, thus, in particular,  we must have had $\pi(v) < \pi(u)$. However, in this case we would have made $u$ a follower of $v$ at this step and set $\ell(u) = v$, which could not have occurred since then $u$ would have been removed from $\ALG$, which completes the proof. 
\end{proof}

\begin{claim}\label{claim:1}
We always have $|\ALG| \leq k+1$ at all time steps.
\end{claim}
\begin{proof}
After the first insertion the result is clear. We demonstrate that the claim holds inductively after each insertion. Only a call to $\ins(v)$ can increase the size of $\ALG$, so consider any such call. If $N(v) \cap \ALG$ is empty, then we add $v$ to $\ALG$, which can possibly increase its size to $k+2$ if it previously had $k+1$ elements. In this case, we remove the element with largest rank and add it to $\cQ$, maintaining the invariant. If there exists a $u^* \in N(v) \cap \ALG$ with smaller rank than $v$, we make $v$ the follower of the vertex in $ N(v) \cap \ALG$ with smallest rank, in which case the size of $\ALG$ is unaffected. In the final case, all points in $N(v) \cap \ALG$ have larger rank than $v$, in which case all of $N(v) \cap \ALG$ (which is non-empty) is made a follower of $v$ and removed from $\ALG$, thereby decreasing or not affecting the size of $\ALG$, which completes the proof. 
\end{proof}

\begin{proposition}[Correctness of Leader Mapping]
At any time step, if $\ALG \leq k$ then every vertex $v \in V$ is either contained in $\ALG$, or has a leader $\ell(v) \in \ALG$ with $(v,\ell(v)) \in E$.
\end{proposition}
\begin{proof}
After processing any update, we first claim that if $\ALG \leq k$ we have $\cQ = \emptyset$. This follows from the fact that the while loop in Line \ref{line:whileMain} of Algorithm \ref{alg:main} does not terminate until one of these two conditions fails to hold. Thus if $\ALG \leq k$, every vertex $v \in V$ is either contained in $\ALG$ (i.e., an active leader), or is a follower of such an active leader, which completes the proof of the proposition, after noting that we only set $\ell(u) = v$ when $(u,v) \in E$ is an edge.
\end{proof}

\begin{lemma}[Correctness of the top-$k$ LFMIS]\label{lem:correctness}
After every time step, we have $\ALG = \LFMIS_{k+1}(G,\pi)$. 
\end{lemma}
\begin{proof}
Order the points in $\ALG = (v_1,\dots,v_r)$ and $\LFMIS_{k+1}(G,\pi) = (u_1,\dots,u_s)$ by rank. We prove inductively that $v_i = u_i$. Firstly, note that $u_1$ is the vertex with minimal rank in $G$. As a result, $u_1$ could not be a follower of any point, since we only set $\ell(u) = v$ when $\pi(v) < \pi(u)$. Thus $u_1$ must either be an inactive leader (as it cannot be equal to $v_j$ for $j > 1$) or an unclustered point. In both cases, one has $u_1 \in \cQ$, which we argue cannot occur. To see this, note that at the end of processing the update, the while loop in  Line \ref{line:whileMain} of Algorithm \ref{alg:main} would necessarily remove $u_1$ from $\cQ$ and insert it. It follows that we must have $u_1 = v_1$.

In general, suppose we have $v_i = u_i$ for all $i \leq j$ for some integer $j < s$. We will prove $v_{j+1} = u_{j+1}$. First suppose $r , s \geq j+1$.
Now by definition of the LSFMIS, the vertex $u_{j+1}$ is the smallest ranked vertex in $V \setminus  \cup_{i \leq j} N(v_i) \cup \{u_i\}$. Since we only set $\ell(u) = v$ when $(u,v) \in E$ is an edge, it follows that $u_{j+1}$ cannot be a follower of $u_{i} = v_i$ for any $i \leq j$. Moreover, since we only set $\ell(u) = v$ when $\pi(v) < \pi(u)$, it follows that $u_{j+1}$ cannot be a follower of $v_{i}$ for any $i > j$, since $\pi(v_i) \geq \pi(u_{j+1})$ for all $i > j$. 
Thus, if $v_{j+1} \neq u_{j+1}$, it follows that either $u_{j+1} \in \cQ$, or $u_{j+1}$ is a follower of some vertex $u' \in \cQ$ with smaller rank than $u_{j+1}$. Then, similarly as above, in both cases the while loop in  Line \ref{line:whileMain} of Algorithm \ref{alg:main} would necessarily remove $u_{j+1}$ (or $u'$ in the latter case) from $\cQ$ and insert it, because $\pi(u_{j+1}) < \pi(v_r)$, and in the latter case if such a $u'$ existed we would have $\pi(u') < \pi(u_{j+1}) < \pi(u_r)$. We conclude that $v_{j+1} = u_{j+1}$.

The only remaining possibility is $r \neq s$. First, if $r > s$, by Claim \ref{claim:1} we have $r \leq k+1$, and by  Proposition \ref{prop:IS} $\ALG$ forms a independent set. Thus $v_1,v_2,\dots,v_s,v_{s+1}$ is an independent set, but since $v_i = u_i$ for $i \leq s$ and $\LFMIS_{k+1}(G,\pi) = \{u_1,\dots,u_s\}$ is a maximal independent set whenever $s \leq k$, this yields a contradiction. Finally, if $r < s$, consider the vertex $u_{r+1}$. Since $u_i = v_i$ for all $i \leq r$, $u_{r+1}$ cannot be a follower of $v_i$ for any $i \in [r]$. As a result, it must be that either $u_{r+1} \in \cQ$ or $u_{r+1}$ is a follower of a vertex in $\cQ$. In both cases, 
at the end of the last update, we had $|\ALG| = r \leq k$ and $\cQ \neq \emptyset$, which cannot occur as the while loop in  Line \ref{line:whileMain} of Algorithm \ref{alg:main} would not have terminated. It follows that $r=s$, which completes the proof.
\end{proof}


\subsection{Amortized Update Time Analysis}\label{sec:amortized}
We now demonstrate that the above algorithm runs in amortized $\tilde{O}(k)$-time per update. We begin by proving a structural result about the behavior of our algorithm. In what follows, let $G^t$ be the state of the graph \textit{after} the $t$-th update. Similarly, let $\ell_t(v) \in V \cup \{\bot\}$ be the value of $\ell(v)$ after the $t$-th update. 

\begin{proposition}\label{prop:orderedPi}
Let $\ins(v_1,\pi(v_1)),\ins(v_2,\pi(v_2)),\dots,\ins(v_r,\pi(v_r))$ be the ordered sequence of calls to the $\ins$ function (Algorithm \ref{alg:ins}) which take places during the processing of any individual update in the stream. Then we have $\pi(v_1) < \pi(v_2) < \cdots < \pi(v_r)$. As a corollary, for any vertex $v$ the function $\ins(v,\pi(v))$ is called at most once per time step. 
\end{proposition}
\begin{proof}
Assume $r>1$, since otherwise the claim is trivial. 
To prove the proposition it will suffice to show two facts: $(1)$ whenever a call to $\ins(v_i,\pi(v_i))$ is made, $\pi(v_i)$ is smaller than the rank of all vertices in the queue $\cQ$, and $(2)$ a call to $\ins(v_i,\pi(v_i))$ can only result in vertices with larger rank being added to $\cQ$. 

To prove $(1)$, note that after the first call to $\ins(v_1,\pi(v_1))$, which may have been triggered directly as a result of $v_1$ being added to the stream at that time step, all subsequent calls to $\ins$ can only be made via the while loop of Line \ref{line:whileMain} in Algorithm \ref{alg:main}, where the point with smallest rank is iteratively removed from $\cQ$ and inserted. Thus, fact $(1)$ trivially holds for all calls to $\ins$ made in this while loop, and it suffices to prove it for  $\ins(v_1,\pi(v_1))$ in the case that $v_1$ is added to the stream at the current update (if $v_1$ was added from the queue, the result is again clear). Now if $\cQ \neq \emptyset$ at the moment $\ins(v_1,\pi(v_1))$ is called, it must be the case that $|\ALG| = k+1$ and $\min_{w \in \cQ} \pi(w) > \max_{w \in \ALG} \pi(w)$ (otherwise the queue would have been emptied at the end of the prior update). Thus, if it were in fact the case that $\pi(v_1) >\min_{w \in \cQ} \pi(w)$, then we also have $\pi(v_1) > \max_{w \in \ALG} \pi(w)$, and therefore the call to $\ins(v_1,\pi(v_1))$ would result in inserting $v_1$ into $\cQ$ on Line \ref{line:firstIfIns} of Algorithm \ref{alg:ins}. Such an update does not modify $\ALG$, and does not change the fact that $\min_{w \in \cQ} \pi(w) > \max_{w \in \ALG} \pi(w)$, thus the processing of the update will terminate after the call to $\ins(v_1,\pi(v_1))$ (contradicting the assumption that $r>1$), which completes the proof of $(1)$.

To prove $(2)$, note that there are only three ways for a call to $\ins(v_i,\pi(v_i))$ to result in a vertex $u$ being added to $\cQ$. In the first case, if $\ell(u)$ was an active leader which was made a follower of $v_i$ as a result of $\ins(v_i,\pi(v_i))$, then we have $\pi(u) < \pi(\ell(u)) < \pi(v_i)$. Next, we could have had $\ell(u) = v_i$ (in the event that $v_i$ was an inactive leader being reinserted from $\cQ$), in which case $\pi(u) < \pi(v_i)$. Finally, it could be the case that $u$ was the active leader in $\ALG$ prior to the call to $\ins(v_i,\pi(v_i))$, and was then removed from $\ALG$ as a result of the size of $\ALG$ exceeding $k+1$ and $u$ having the largest rank in $\ALG$. This can only occur if $v_i$ was added to $\ALG$ and had smaller rank than $u$, which completes the proof of $(2)$.

Since by $(1)$ every time $\ins(v_i,\pi(v_i))$ is called $\pi(v_i)$ is smaller than the rank of all vertices in the queue, and by $(2)$ the rank of all new vertices added to the queue as a result of $\ins(v_i,\pi(v_i))$ will continue to be larger than $\pi(v_i)$, it follows that $v_{i+1}$, which by construction must be the vertex with smallest rank in $\cQ$ after the call to $\ins(v_i,\pi(v_i))$, must have strictly larger rank than $v_i$, which completes the proof of the proposition. 
\end{proof}

The following proposition is more or less immediate. It implies, in particular, that a point can only be added to $\cQ$ once per time step (similarly, $v$ can be removed from $\cQ$ once per time step). 

\begin{proposition}\label{prop:afterQ}
Whenever a vertex $v$ in the queue $\cQ$ is removed and $\ins(v,\pi(v))$ is called, the vertex $v$ either becomes a follower of an active leader, or an active leader itself. 
\end{proposition}
\begin{proof}
If $v$ shares an edge with a vertex in $\ALG$ with smaller rank, it becomes a follower of such a vertex. Otherwise, all vertices in $N(v) \cap \ALG$ become followers of $v$, and $v$ becomes an active leader by construction (possibly resulting in an active leader of larger rank to be removed from $\ALG$ as a result of it no longer being contained in $\LFMIS_{k+1}$).
\end{proof}

Equipped with the prior structural propositions, our approach for bounding the amortized update time is to first observe that, on any time step $t$, our algorithm only attempts to insert a vertex~$v$, thereby spending $O(k)$ time to search for edges between $v$ and all members of $\ALG$, if either $v$ was the actual vertex added to the stream on step $t$, or when $v$ was added to $\cQ$ on step $t$ or before. Thus, it will suffice to bound the total number of vertices which are ever added into $\cQ$ --- by paying a cost $O(k + \log n)$ for each vertex $v$ which is added to the queue, we can afford both the initial $O(\log n)$ cost of adding it to the priority queue, as well as the $O(k)$ runtime cost of possibly later reinserting $v$ during the while loop in Line \ref{line:whileMain} of Algorithm \ref{alg:main}. We formalize this in the following proposition.

\begin{proposition}\label{prop:T}
Let $T$ be the total number of times that a vertex is inserted into the queue $\cQ$ over the entire execution of the algorithm, where two insertions of the same vertex $v$ on separate time steps are counted as distinct insertions.  Then the total runtime of the algorithm, over a sequence of $M$ insertions and deletions, is at most $O(T (k + \log n) + Mk)$, where $n$ is the maximum number of vertices active at any given time.
\end{proposition}
\begin{proof}
Note that the only actions taken by the algorithm consist of adding and removing vertices $v$ from $\cQ$ (modifying the value of $\ell(v)$ in the process, and possibly deleting $\cF_v$), and computing $\ALG \cap N(v)$ for some vertex $v$. The latter requires $O(k)$ time since we have $|\ALG| \leq k+1$ at all time steps. Given $O(\log n)$ time to insert or query from a priority queue with at most $n$ items, we have that $O(T \log n)$ upper bounds the cost of all insertions and deletions of points to $\cQ$. Moreover, all calls to compute $\ALG \cap N(v)$ for some vertex $v$ either occur when $v$ is the vertex added to the stream on that time step (of which there is at most one), or when $v$ is inserted after previously having been in $\cQ$. By paying each vertex $v$ a sum of $O(k)$ when it is added to $\cQ$, and paying $O(k)$ to each vertex when it is first added to the stream, it can afford the cost of later computing $\ALG \cap N(v)$  when it is removed. This results in a total cost of $O(T(k+\log n) + Mk)$, which completes the proof. 
\end{proof}

In what follows, we focus on bounding the quantity $T$. To accomplish this, observe that a vertex $v$ can be added to $\cQ$ on a given time step $t$ in one of three ways:
\begin{figure}[H]
    \centering
  
\begin{Frame}[Scenarios where $v$ is added to $\cQ$]
    \begin{enumerate}
\item The vertex $v$ was added in the stream on time step $t$. In this case, $v$ is added to $\cQ$ when the if statement on Line \ref{line:firstIfIns} of Algorithm \ref{alg:ins} executes. 
\item The vertex $v$ is added to $\cQ$ when it was previously led by $\ell_{t-1}(v) \in V$, and either $\ell_{t-1}(v)$ becomes a follower of another leader during time step $t$ (resulting in $v$ being added to $\cQ$), or $\ell_{t-1}(v)$ is deleted. This can occur in either Lines \ref{line:queue1Ins} or \ref{line:queue2Ins}  of Algorithm \ref{alg:ins} for the first case, or in Lines \ref{line:queue1Del} or \ref{line:queue2Del} of Algorithm \ref{alg:del} in the case of  $\ell_{t-1}(v)$  being deleted. 
\item The vertex $v$ was previously in $\LFMIS_{k+1}$, and subsequently left $\LFMIS_{k+1}$ because $|\LFMIS_{k+1}| = k+1$ and a new vertex $u$ was added to $\LFMIS_{k+1}$ with smaller rank. This occurs in Line \ref{line:delOverflow} of Algorithm \ref{alg:ins}. 
    \end{enumerate}
\end{Frame}

\end{figure}

Obviously, the first case can occur at most once per stream update, so we will focus on bounding the latter two types of additions to $\cQ$. For any step $t$, define $\cA_\pi^t$ to be the number of vertices that are added to $\cQ$ as a result of the second form of insertions above. Namely, $\cA_\pi^t = |\{v \in G^{t} : \ell_{t-1}(v) \in V, \text{ and } \ell_t(v) \neq \ell_{t-1}(v) \}|$. Next, define $\cB_\pi^t$ to be the number of leaders which were removed from $\ALG$ Line \ref{line:delOverflow} of Algorithm \ref{alg:ins} (i.e., insertions into $\cQ$ of the third kind above). Letting $T$ be as in Proposition \ref{prop:T}, we have $T \leq  M+  \sum_t \cA_\pi^t + \cB_\pi^t $.

To handle $\sum_t \cA_\pi^t$ and $\sum_t \cB_\pi^t$, we demonstrate that each quantity can be bounded by the total number of times that the \textit{eliminator} of a vertex changes. Recall from Section \ref{sec:kCenters} that, given a graph $G=(V,E)$, $v \in V$, and ranking $\pi:V \to [0,1]$, the eliminator of $v$, denoted $\elim_{G,\pi}(v)$, is defined as the vertex of smallest rank in the set $(N(v) \cup \{v\} ) \cap \LFMIS(G,\pi)$. Now define $\cC_\pi^t$ to be the number of vertices whose eliminator changes after time step $t$. Formally, for any two graphs $G,G'$ differing in at most once vertex, we define 
$\cC_\pi(G,G') = \{v \in V | \elim_{G,\pi}(v) \neq \elim_{G',\pi}(v)\}$, and set $\cC_\pi^t = |\cC_\pi(G^{t-1},G^{t})|$. 
We now demonstrate that $\sum_t \cC_\pi^t$ deterministically upper bounds both $\sum_t \cA_\pi^t$ and $\sum_t \cB_\pi^t$.

\begin{lemma}\label{lem:main}
Fix any ranking $\pi: V \to [0,1]$. Then we have $\sum_t \cA_\pi^t \leq 5 \sum_t \cC_\pi^t$, and moreover for any time step $t$ we have $\cB_\pi^t \leq \cC_\pi^t$.
\end{lemma}
\begin{proof}
We first prove the second statement. Fix any time step $t$, and let $v_1,\dots,v_r$ be the $r = |\cB_\pi^t|$ vertices which were removed from $\ALG$, ordered by the order in which they were removed from $\ALG$. For this to occur, we must have inserted at least $r$ vertices $u_1,\dots,u_r$ into $\ALG$ which were not previously in $\ALG$ on the prior step; in fact, Line \ref{line:delOverflow} of Algorithm \ref{alg:ins} induces a unique mapping from each $v_i$ to the vertex $u_i$ which forced it out of $\ALG$ during a call to $\ins(u_i,\pi(u_i))$. Note that, under this association, we have $\pi(u_i) < \pi(v_i)$ for each $i$. We claim that the eliminator of each such $u_i$ changed on time step $t$.

Now note that $\{u_1,\dots,u_r\}$ and $\{v_1,\dots,v_r\}$ are disjoint, since $u_i$ was inserted before $u_{i+1}$ during time step $t$ by the definition of the ordering, and so $\pi(u_1) < \pi(u_2) < \dots < \pi(u_r)$ by Proposition \ref{prop:orderedPi}, so no $u_i$ could be later kicked out of $\ALG$ by some $u_j$ with $j > i$. It follows that none of $u_1,\dots,u_r$ were contained in $\LFMIS_{k+1}(G^{t-1},\pi)$, but they are all in $\LFMIS_{k+1}(G^{t},\pi)$. Now note that it could not have been the case that $u_i \in \LFMIS(G^{t-1},\pi)$, since we had $v_i \in \LFMIS_{k+1}(G^{t-1},\pi)$ but $\pi(u_i) < \pi(v_i)$. Thus $u_i \notin \LFMIS(G^{t-1},\pi)$, and therefore the eliminator of $u_i$ changed on step $t$ from $\elim_{G^{t-1},\pi}(u_i) \neq u_i$ to $\elim_{G^t,\pi}(u_i) = u_i$, which completes the proof of the second statement.

We now prove the first claim that $\sum_t \cA_\pi^t \leq \sum_t \cC_\pi^t$.  Because a vertex can be inserted into $\cQ$ at most once per time step (due to Proposition \ref{prop:afterQ}), each insertion $\cQ$ which contributes to $\sum_t \cA_\pi^t$ can be described as a  vertex-time step pair $(v,t)$, where we have $\ell_{t}(v) \neq \ell_{t-1}(v) \in V$ because either $\ell_{t-1}(v)$ became a follower of a vertex in $\LFMIS_{k+1}(G^t,\pi)$, or because $\ell_{t-1}(v)$ was deleted on time step $t$. 
We will now need two technical claims.

\begin{claim}\label{claim:2}
Consider any vertex-time step pair $(v,t)$ where $\ell_t(v) \in V$ and $\ell_{t-1}(v) \neq \ell_t(v)$. In other words, $v$ was made a follower of some vertex $\ell_t(v)$ during time step $t$. Then $\ell_t(v) = \elim_{G^t,\pi}(v)$.
\end{claim}
\begin{proof}
First note that the two statements of the claim are equivalent, since if $\ell(v)$ is set to $u \in V$ during time step $t$, then by Proposition \ref{prop:orderedPi} we have that $\ell(v)$ is not modified again during the processing of update $t$, so $u = \ell_{t}(v)$. Now the algorithm would only set $\ell_t(v) = u$ in one of two cases. In this first case, it occurs during a call to $\ins(v,\pi(v))$, in which case $\ell_t(v)$ is set to the vertex with smallest rank in $\LFMIS_{k+1}(G^t,\pi)\cap N(v)$, which by definition is $\elim_{G^t,\pi}(v)$. In the second case, $v$ was previously in $\LFMIS_{k+1}(G^{t-1},\pi)$, and $\ell(v)$ was changed to a vertex $w$ during a call to  $\ins(w,\pi(w))$, where $w \in N(v)$ and $\pi(w) < \pi(v)$. Since prior to this insertion $v$ was not a neighbor of any point in $\LFMIS_{k+1}(G^{t-1},\pi)$, and since by Proposition \ref{prop:orderedPi} the $\ins$ function will not be called again on time $t$ for a vertex with rank smaller than $w$, it follows that $w$ has the minimum rank of all neighbors of $v$ in $\LFMIS(G^{t},\pi)$, which completes the claim. 
\end{proof}

\begin{claim}\label{claim:3}
Consider any vertex-time step pair $(v,t)$ where $\ell_t(v) = \bot$ and $\ell_{t-1}(v) = \elim_{G^{t-1},\pi}(v) $. Then $\elim_{G^{t-1},\pi}(v) \neq \elim_{G^{t},\pi}(v) $.
\end{claim}
\begin{proof}
If $\ell_{t-1}(v) = \elim_{G^{t-1},\pi}(v)$, then $\elim_{G^{t-1},\pi}(v) \in v$  and $\ell_t(v)$ is changed to $\bot$ during time step $t$, then as in the prior claim, this can only occur if $\ell_{t-1}(v)$ is made a follower of another point in $\LFMIS_{k+1}(G^t,\pi)$, or if $\ell_{t-1}(v)$ is deleted on that time step. In both cases we have $\ell_{t-1}(v) \notin \LFMIS(G^t,\pi)$. Since $\elim_{G^{t},\pi}(v)$ is always in $\LFMIS(G^t,\pi)$, the claim follows.  
\end{proof}

Now fix any vertex $v$, and let $\sigma_1,\dots,\sigma_M$ be the sequence of eliminators of $v$, namely $\sigma_t = \elim_{G^t,\pi}(v)$ (note that $\sigma_t$ is either a vertex in $V$ or $\sigma_t = \emptyset$). Similarly define $\lambda_1,\dots,\lambda_M$ by $\lambda_t = \ell_t(v)$, and note that $\lambda_t \in V \cup \{\bot\}$. To summarize the prior two claims: each time $\lambda_{t-1} \neq \lambda_t$ and $\lambda_t \in V$, we have $\sigma_t = \lambda_t$; namely, the sequences become aligned at time step $t$. Moreover, whenever the two sequences are aligned at some time step $t$, namely $\sigma_t = \lambda_t$, and subsequently $\lambda_{t+1} = \bot$, we have that $\sigma_{t+1} \neq \sigma_t$. We now prove that every five subsequent changes in the value of $\lambda$ cause at least one unique change in $\sigma$.

To see this, let $t_1 < t_2 < t_3$ be three subsequent changes, so that $\lambda_{t_1} \neq \lambda_{t_1 - 1}$, $\lambda_{t_2} \neq \lambda_{t_2 - 1}$, $\lambda_{t_3} \neq \lambda_{t_3 - 1}$, and $\lambda_i$ does not change for all $i=t_1,\dots,t_2-1$ and $i= t_2,\dots,t_3-1$. First, if $\lambda_{t_1},\lambda_{t_2} \in V$, by Claim \ref{claim:2} we have $\sigma_{t_1} = \lambda_{t_1}$ and  $\sigma_{t_2} = \lambda_{t_2}$, and thus $\sigma_{t_1} \neq \sigma_{t_2}$, so $\sigma$ changes in the interval $[t_1,t_2]$. If $\lambda_{t_1} \in V$ and $\lambda_{t_2} = \bot$, we have $\sigma_{t_1} = \lambda_{t_1}$, and so if $\sigma$ does not change by time $t_2-1$ we have $\sigma_{t_2-1} = \lambda_{t_2-1}$, and thus $\sigma_{t_2} \neq \sigma_{t_2-1}$ by Claim \ref{claim:3}, so $\sigma$ changes in the interval $[t_1,t_2]$. Finally, if $\lambda_{t_1} = \bot$, then we must have $\lambda_{t_2} \in V$, and so $\lambda_{t_3} = \bot$. Then by the prior argument, $\sigma$ must change in the interval $[t_2,t_3]$. Thus, in each case, $\sigma$ must change in the interval $[t_1,t_3]$. To avoid double counting changes which occur on the boundary, letting $t_1,\dots,t_r$ be the sequence of all changes in $\lambda$, it follows that there is at least one change in $\sigma$ in each of the disjoint intervals $(t_{5i+1}, t_{5(i+1)})$ for $i=0,1,2,\dots,\lfloor r/5 \rfloor$. It follows that $\sum_t \cA_\pi^t \leq 5 \sum_t \cC_\pi^t$, which completes the proof. 
\end{proof}

The following theorem, due to \cite{behnezhad2019fully}, bounds the expected number of changes of eliminators which occur when a vertex is entirely removed or added to a graph.


\begin{theorem}[Theorem 3 of \cite{behnezhad2019fully}]\label{thm:elimBound}
Let $G = (V,E)$ be any graph on $n$ vertices, and let $G' = (V',E')$ be obtained by removing a single vertex from $V$ along with all incident edges. Let $\pi: V \to [0,1]$ be a random mapping. Let $\cC_\pi(G,G') = \{v \in V | \elim_{G,\pi}(v) \neq \elim_{G',\pi}(v)\}$. Then we have $\mathbb{E}_\pi \left[|\cC_\pi(G,G')|\right] = O(\log n)$.
\end{theorem}

\begin{theorem} \label{thm:LFMISMain} There is a algorithm which, on a fully dynamic stream of insertions and deletions of vertices to a graph $G$, maintains at all time steps a top-$k$ LFMIS of $G$ with leaders (Definition \ref{def:LFMISLead}) under a random ranking $\pi: V \to [0,1]$. The expected amortized per-update time of the algorithm is $O(k \log n + \log^2 n)$, where $n$ is the maximum number active of vertices at any time. Moreover, the algorithm does not need to know $n$ in advance. 
\end{theorem}
\begin{proof}
By the above discussion, letting $T$ be as in Proposition \ref{prop:T}, we have $T \leq  M+  \sum_t \cA_\pi^t + \cB_\pi^t $. By the same proposition, the total update time of the algorithm over a sequence of $n$ updates is at most $O(T(k + \log n) + kM)$.  By Lemma \ref{lem:main}, we have $\sum_t \cA_\pi^t + \cB_\pi^t  \leq 6 \sum_t \cC^t_\pi$, and by Theorem \ref{thm:elimBound} we have $\mathbb{E}_\pi \left[\sum_t \cC^t_\pi\right] = O(M \log n)$. It follows that $\ex{T} = O(M \log n)$, therefore the total update time is $O(kM \log n + M \log^2 n)$, which completes the proof. 
\end{proof}

\section{Fully Dynamic $k$-Centers via Locally Sensitive Hashing}\label{sec:LSH}

In this section, we demonstrate how the algorithm for general metric spaces of Section \ref{sec:generalMetric} can be improved to run in \textit{sublinear} in $n$ amortized update time, even when $k = \Theta(n)$, if the metric in question admits good locally sensitive hash functions (introduced below in Section \ref{sec:LSHAlg}). Roughly speaking, a locally sensitive hash function is a mapping $h: \cX \to U$, for an universe $U$, which has the property that points which are close in the metric should collide, and points which are far should not. Thus, by when searching for points which are close to a given $x \in \cX$, one can first apply a LSH to quickly prune far points, and search only through the points in the hash bucket $h(x)$. We will use this approach to speed up the algorithm from Section \ref{sec:generalMetric}. 

\paragraph{Summary of the LSH-Based Algorithm. }
We now describe the approach of our algorithm for LSH spaces.
Specifically, first note that the factor of $k$ in the amortized update time in Proposition \ref{prop:T} comes from the time required to compute $S = \ALG \cap N(v)$ in Algorithm \ref{alg:ins}. However, to determine which of the three cases we are in for the execution of Algorithm \ref{alg:ins}, we need only be given the value $u^* = \arg \min_{u' \in S} \pi(u')$ of the vertex in $S$ with smallest rank, or $S = \emptyset$ if none exists. If $S = \emptyset$, then the remainder of Algorithm \ref{alg:ins} runs in constant time. If $\pi(u^*) < \pi(v)$, where $v$ is the query point, then the remaining run-time is constant unless $v$ was a leader, in which case it is proportional to the number of followers of $v$, each of which are inserted into $\cQ$ at that time. Lastly, if $\pi(u^*) > \pi(v)$, then we search through each $u \in S$, make it a follower of $v$, and add the followers of $u$ to $\cQ$. Since if $u$ was previously in $\ALG$, it must have also been in $\LFMIS_{k+1}$ on the prior time step, it follows that each such $u \in S$ changes its eliminator on this step.

In summary, after the computation of $S = \ALG \cap N(v)$, the remaining runtime is bounded by the sum of the number of points added to $\cQ$, and the number of points that change their eliminator on that step. Since, ultimately, the approach in Section \ref{sec:amortized} was to bound $T$ by the total number of times a point's eliminator changes, our goal will be to obtain a more efficient data structure for returning $S = \ALG \cap N(v)$. Specifically, if after a small upfront runtime $R$, such a data structure can read off the entries $S$ in the order of the rank, each in constant time, one could therefore replace the factors of $k$ in Proposition \ref{prop:T} by $R$. We begin by formalizing the guarantee that such a data structure should have. 

We will demonstrate that approximate nearest neighbor search algorithms based on locally sensitive hashing can be modified to have the above properties. However, since such an algorithm will only be approximate, it will sometimes returns points in $S$ which are farther than distance $r$ from $v$, where $r$ is the threshold. Thus, the resulting graph defined by the locally sensitive hashing procedure will now be an \textit{approximate threshold graph:}

\begin{definition}\label{def:approxThreshGraph}
Fix a point set $P$ from a metric space $(\cX,d)$, and real values $r>0$ and $c\geq 1$. A $(r,c,L)$-approximate threshold graph $G_{r,c} = (V(G_{r,c}),E(G_{r,c}))$ for $P$ is any graph with $V(G_{r,c}) = P$, and whose whose edges satisfy $E(G_r) \subseteq E(G_{r,c})$ and $|E(G_{r,c}) \setminus E(G_{cr})| \leq L$, where $E(G_r), E(G_{cr})$ are the edge set of the threshold graphs $G_r,G_{cr}$ respectively. 

\end{definition}

If $L = 0$, it is straightforward to see that an algorithm for solving the top-$k$ LFMIS problem on a $(r,c,0)$-approximate threshold graph $G_{r,c}$ can be used to obtain a $c(2+\eps)$ approximation to $k$-centers. When $L>0$, an algorithm can first check, for each edge $e \in E(G_{r,c})$ it considers, whether $e \in G_{rc}$, and discard it if it is not the case. We will see that the runtime of handling a $(r,c,L)$-approximate threshold graph will depend linearly on $L$. Moreover, we will set parameters so that $L$ is a constant in expectation. 


\subsection{Locally Sensitive Hashing and the LSH Algorithm}\label{sec:LSHAlg}

We begin by introducing the standard definition of a locally sensitive hash family for a metric space \cite{indyk1998approximate}.

\begin{definition}[Locally sensitive hashing \cite{indyk1998approximate,har2012approximate}]\label{def:LSH}
Let $\cX$ be a metric space, let $U$ be a range space, and let $r \geq 0$ $c \geq 1$ and $0 \leq p_2 \leq p_1  \leq 1$ be reals.
A family $\cH = \{h : \cX \to U\}$ is called $(r,cr,p_1,p_2)$-sensitive if for any $q,p \in \cX$:

\begin{itemize}
    \item If $d(p,q) \leq r$, then $\prb{\cH}{h(q) = h(p)} \geq p_1$. 
    \item If $d(p,q) > cr$, then $\prb{\cH}{h(q) = h(p)} \leq p_2$. 
\end{itemize}
\end{definition}

Given a $(r,cr,p_1,p_2)$-sensitive family $\cH$, we can define $\cH^t$ to be the set of all functions $h^t: \cX \to U^t$ defined by $h^t(x) = (h_1(x),h_2(x),\dots,h_t(x))$, where  $h_1,\dots,h_k \in \cH$. In other words, a random function from $\cH^t$ is obtained by drawing $t$ independent hash functions from $\cH$ and concatenating the results. It is easy to see that the resulting hash family $\cH^t$ is $(r,cr,p_1^t,p_2^t)$-sensitive. 
We now demonstrate how a $(r,c)$-approximate threshold graph can be defined via a locally sensitive hash function. 

\begin{definition}\label{def:inducedGraph}
Fix a metric space $(\cX,d)$ and a finite point set $P \subset \cX$, as well as integers $t,s \geq 1$. Let $\cH: \cX \to U$ be a $(r,cr,p_1,p_2)$-sensitive family. Then a graph $G_{r,cr}(P,\cH,t,s) = ( V,E)$ induced by $\cH$ is a random graph which is generated via the following procedure. First, one randomly selects hash functions $h_1,h_2,\dots,h_s \sim \cH^t$. Then the vertex set is given by $V=P$, and then edges are defined via $(x,y) \in E$ if and only if $h_i(x) = h_i(y)$ for some $i \in [s]$. 
\end{definition}

We now demonstrate that, if $\cH$ is a sufficiently sensitive hash family, the random graph $G_{r,cr}(P,\cH,t,s)$ constitutes a $(r,c,L)$-approximate threshold graph with good probability, where $L$ is a constant in expectation. 

\begin{proposition}\label{prop:isAThreshold}
Fix a metric space $(\cX,d)$ and a point set $P \subset \cX$ of size $|P|= n$, and let $\cH: \cX \to U$ be a $(r,cr,p_1,p_2)$-sensitive family. Fix any $\delta \in (0,\frac{1}{2})$. Set $s = \ln(n^2/\delta) n^{2\rho}/p_1 $, where $\rho = \ln \frac{1}{p_1}/\ln \frac{1}{p_2}$, and  $t = \lceil 2 \log_{1/p_2} n  \rceil$. Then, with probability at least $1-\delta$, the the random graph $G_{r,cr}(P,\cH,t,s)$ is a $(r,c,L)$-approximate threshold graph, where $L$ is a random variable satisfying $\ex{L} < 2$ (Definition \ref{def:approxThreshGraph}). 
\end{proposition}
\begin{proof}

First, fix any $x,y \in P$ such that $d(x,y) > c r$. We have that $\prb{h \sim \cH^k}{h(x) = h(y)} \leq p_2^t < \frac{1}{ n^2}$. It follows that 
\[\exx{h_1,\dots,h_s \sim \cH^t}{|E(G_{r,cr}(P,\cH,t,s) \setminus E(G_{r })|} \leq \sum_{(x,y) \in P^2} \frac{1}{n^2} < 1\]
Namely, we have $\ex{L} < 1$ where $L = |E(G_{r,cr}(P,\cH,t,s) \setminus E(G_{r })|$.
Next, fix any $(x,y) \in E(G_r)$. We have 
\[\prb{h \sim \cH^k}{h(x) = h(y)} \geq p_1^t > p_1^{2 \log_{1/p_2} n  + 1}= p_1 (n^2)^{-\rho}\] Thus, the probability that at least one $h_i$ satisfies $h_i(x) = h_i(y)$ is at least
\[1-(1- p_1 n^{-2\rho})^s >1-(1/e)^{\ln(n^2/\delta)} =  1- \delta/n^2\] 
After a union bound over all such  possible pairs, it follows that $(x,y) \in E(G_{r,cr}(P,\cH,t,s)) $ for all $(x,y) \in E(G_r)$ with probability at least $1-\delta$. Note that since $\delta < 1/2$ and $L$ is a non-negative random variable, it follows that conditioning on the prior event can increase the expectation of $L$ by at most a factor of $2$, which completes the proof. 
\end{proof}

We will now describe a data structure which allows us to maintain a subset $\cL$ of vertices of the point set $P$, and quickly answer queries for neighboring edges of a vertex $v$ in the graph $G$ defined by the intersection of $G_{r,cr}(P,\cH,t,s)$ and $G_{cr}$. Note that if $G_{r,cr}(P,\cH,t,s)$ is a $(r,c,L)$-approximate threshold graph, then this intersection graph $G$ satisfies $G_r \subseteq G \subseteq G_{cr}$, and is therefore a  $(r,c,0)$-approximate threshold graph. It is precisely this graph $G$ which we will run our algorithm for top-$k$ LFMIS on. However, in addition to finding all neighbors of $v$ in $\cL$, we will also need to quickly return the neighbor with smallest rank $\pi$, where $\pi:V \to [0,1]$ is a random ranking as in Section \ref{sec:generalMetric}. Roughly, the data structure will hash all points in $\cL$ into the hash buckets given by $h_1,\dots,h_s$, and maintain each hash bucket via a binary search tree of depth $O(\log n)$, where the ordering is based on the ranking $\pi$. 

For the following Lemma and Theorem, we fix a metric space $(\cX,d)$ and a point set $P \subset \cX$ of size $|P|= n$, as well as a scale $r > 0$, and approximation factor $c$. Moreover, let $\cH: \cX \to U$ be a $(r,cr,p_1,p_2)$-sensitive hash family for the metric space $\cX$, and let $\Run(\cH)$ be the time required to evaluate a hash function $h \in \cH$. Furthermore, let $\pi: P \to [0,1]$ be any ranking over the points $P$, such that $\pi(x)$ is truncated to $O(\log n)$ bits, and such that $\pi(x) \neq \pi(y)$ for any distinct $x,y \in P$ (after truncation). Note that the latter holds with probability $1-1/\poly(n)$ if $\pi$ is chosen uniformly at random. Lastly, for a graph $G = (V(G),E(G))$, let $N_G(v)$ be the neighborhood of $v$ in $G$ (in order to avoid confusion when multiple graphs are present). 
\begin{lemma}\label{lem:binaryDataStructure}
Let $G_{r,cr}(P,\cH,t,s)$ be a draw of the random graph as in Definition \ref{def:inducedGraph}, where $r,c,P,\cH$ are as above, such that $G_{r,cr}(P,\cH,t,s)$ is a $(r,c,L)$-approximate threshold graph. Let $G = G_{r,cr}(P,\cH,t,s)\cap G_{cr}$. Then there is a fully dynamic data structure, which maintains a subset $\cL \subset V(G)$ and can perform the following operations:
\begin{itemize}
    \item $\ttx{Insert}(v)$: inserts a vertex $v$ into $\cL$ in time $O(s  \log n + ts \Run(\cH))$
    \item $\ttx{Delete}(v)$: deletes a vertex $v$ from $\cL$ in time $O(s \log n + ts \Run(\cH))$
    \item $\ttx{Query-Top}(v)$: returns $u^* = \arg \min_{u \in N_G \cap \cL} \pi(u)$, or \textsc{Empty} if $N_G(v) \cap \cL = \emptyset$, in time $O(sL \log n + ts \Run(\cH)) $
    \item $\ttx{Query-All}(v)$: returns the set $N_G(v) \cap \cL$, running in time $O(s(|N_G(v) \cap \cL| + L) \log n +  ts \Run(\cH))$. 
\end{itemize}
Moreover, given the hash functions $h_1,\dots,h_s \in \cH^t$ which define the graph $G_{r,cr}(P,\cH,t,s)$, the algorithm is deterministic, and therefore correct even against an adaptive adversary. 
\end{lemma}
\begin{proof}
For each $i \in [s]$ and hash bucket $b \in U$, we store a binary tree $T_{i,b}$ with depth at most $O(\log n)$, such that each node $z$ corresponds to an interval $[a,b] \subset [0,1]$, and the left and right children of $z$ correspond to the intervals $[a,(a+b)/2]$ and $[(a+b/2),b]$, respectively. Moreover, each node $z$ maintains a counter for the number of points in $\cL$ which are stored in its subtree. Given a vertex $v$ with rank $\pi(v)$, one can then insert $v$ into the unique leaf of $T_{i,b}$ corresponding to the $O(\log n)$ bit value $\pi(v)$ in time $O(\log n)$. Note that, since the keys $\pi(v)$ are unique, each leaf contains at most one vertex. Similarly, one can remove and search for a vertex from $T_{i,b}$ in time $O(\log n)$.

When processing any query for an input vertex $v \in P$, one first evaluates all $ts$ hash functions required to compute $h_1(v),\dots,h_s(v)$, which requires  $ts \Run(\cH)$ time. For insertions and deletions, one can insert $v$ from each of the $s$ resulting trees in time $O(\log n)$ per tree, which yields the bounds for $\ttx{Insert}(v)$ and $\ttx{Delete}(v)$. For $\ttx{Query-Top}(v)$, for each $i \in [s]$, one performs an in-order traversal of $T_{i,h_i(v)}$, ignoring nodes without any points stored in their subtree, and returns the first leaf corresponding to a vertex $u \in N_G(v)$, or \textsc{Empty} if all points in the tree are examined before finding such a neighbor. Each subsequent non-empty leaf in the traversal can be obtained in $O(\log n)$ time, and since by definition of a $(r,c,L)$-approximate threshold graph, $h_i(v) = h_i(u')$ for at most $u'$ vertices with $(v,u') \notin G_{cr}$, it follows that one must examine at most $L$ vertices $u'$ with $(v,u') \notin G_{cr}$ before one finds  $u^* = \arg \min_{u \in N_G \cap \cL} \pi(u)$ (or exhausts all points in the tree). Thus, the runtime is $O(L \log n)$ to search through each of the $s$ hash functions, which results in the desired bounds. 

Finally, for $\ttx{Query-All}(v)$, one performs the same search as above, but instead completes the full in-order traversal of each tree $T_{i,h_i(v)}$. By the  $(r,c,L)$-approximate threshold graph property, each tree $T_{i,h_i(v)}$ contains at most $|N_G(v) \cap \cL| + L$ vertices from $\cL$, after which the runtime follows by the argument in the prior paragraph. 
\end{proof}

Given the data structure from Lemma \ref{lem:binaryDataStructure}, we will now demonstrate how the algorithm from Section \ref{sec:generalMetric} can be implemented in sublinear in $k$ time, given a sufficently good LSH function for the metric. The following theorem summarizes the main consequences of this implementation, assuming the graph $G_{r,cr}(P,\cH,t,s)$ is a $(r,c,L)$-approximate threshold graph. 

\begin{theorem}\label{thm:LSHPreMain}
 Let $G_{r,cr}(P,\cH,t,s)$ be a draw of the random graph as in Definition \ref{def:inducedGraph}, where $r,c,P,\cH$ are as above, such that $G_{r,cr}(P,\cH,t,s)$ is a $(r,c,L)$-approximate threshold graph. Let $G = (V,E)$ be the graph with $V = P$ and $E  = E(G_{r,cr}(P,\cH,t,s))\cap E(G_{cr})$. Then there is a fully dynamic data structure which, under a sequence of vertex insertions and deletions from $G$, maintains a top-$n$ LFMIS with leaders (Definition \ref{def:LFMISLead}) of $G$ at all time steps. The expected amortized per-update runtime of the algorithm is $O( ( sL \log n + ts \Run(\cH)) \log n + \log^2 n )$, where the expectation is taken over the choice of $\pi$.
\end{theorem}
\begin{proof}
The algorithm is straightforward: we run the fully dynamic algorithm for top-$k$ LFMIS with leaders from Section \ref{sec:generalMetric}, however we utilize the data structure from Lemma \ref{lem:binaryDataStructure} to compute $S = \cL \cap N_G(v)$ in Algorithm \ref{alg:ins} (where $\cL = \cL_{n+1}$), as well as handle deletions from $\cL$ in Algorithm \ref{alg:del}. Note that to handle a call to $\ins(v)$ of Algorithm \ref{alg:ins}, one first calls  $\ttx{Query-Top}(v)$ in data structure from Lemma \ref{lem:binaryDataStructure}. If the result is $\emptyset$, or $u$ with $\pi(u) < \pi(v)$, then one can proceed as in Algorithm \ref{alg:ins} but by updating $\cL$ via Lemma \ref{lem:binaryDataStructure}. If the result is $u$ with $\pi(u) > \pi(v)$, one then calls  $\ttx{Query-All}(v)$ to obtain the entire set $S = N_G(v) \cap \cL$, each of which will subsequently be made a follower of $v$. 

Let $T$ be the total number of times that a vertex is inserted into the queue $\cQ$ over the entire execution of the algorithm (as in Proposition \ref{prop:T}), and as in 
Section \ref{sec:amortized}, we let $\cC^t_\pi$ denote the number of vertices whose eliminator changed after the $t$-th time step. We first prove the following claim, which is analogous to  Proposition \ref{prop:T}. 

\begin{claim}
The total runtime of the algorithm, over a sequence of $M$ insertions and deletions of vertices from $G$, is at most $O(T(\lambda + \log n) + \lambda (M+ \sum_{t \in [M}) \cC^t_\pi))$, where $\lambda = sL \log n + ts \Run(\cH)$. 
\end{claim}
\begin{proof}
First note that $\lambda$ upper bounds the cost of inserting and deleting from $\cL$, as well as calling $\ttx{Query-Top}(v)$. For every vertex, when it is first inserted into the stream, we pay it a cost of $\lambda$ to cover the call to $\ttx{Query-Top}(v)$. Moreover, whenever a vertex is added to the queue $\cQ$, we pay a cost of $\lambda$ to cover a subsequent call to $\ttx{Query-Top}(v)$ when it is removed from the queue and inserted again, plus an additional $O(\log n)$ required to insert and remove the top of a priority queue. The only cost of the algorithm which the above does not cover is the cost of calling $\ttx{Query-All}(v)$, which can be bounded by $O(\lambda \cdot |N_G(v) \cap \cL| )$. Note that, by correctness of the top-$k$ LFMIS algorithm (Lemma \ref{lem:correctness}), each vertex in $\cL$ is its own eliminator at the beginning of each time step. It follows that each vertex $u \in N_G(v) \cap \cL$ had its eliminator changed on step $t$, since $\ttx{Query-All}(v)$ is only called on time step $t$ in the third case of Algorithm \ref{alg:ins}, where all points in $|N_G(v) \cap \cL|$ will be made followers of $v$. Thus the total cost of all calls to $\ttx{Query-All}(v)$ can be bounded by $\lambda \sum_{t \in [M}) \cC^t_\pi$, which completes the proof of the claim. 
\end{proof}

Given the above, by Lemma \ref{lem:main} we have that $T \leq M \leq 6 \cC^t_\pi$, and by Theorem \ref{thm:elimBound} we have $\exx{\pi}{\sum_t \cC^t_\pi} = O(M \log n)$. It follows that, the expected total runtime of the algorithm, taken over the randomness used to generate $\pi$ (with $h_1,\dots,h_s$ previously fixed and conditioned on) is at most $O(M \log^2 n + M \lambda \log n)$ as needed.

\end{proof}

\begin{theorem}\label{thm:lshMain}
Let $(\cX,d)$ be a metric space, and fix $\delta \in (0,1/2)$. Suppose that for any $r \in (r_{\min}, r_{\max})$ there exists an $(r,cr,p_1,p_2)$-sensitive hash family $\cH_r: \cX \to U$, such that each $h \in \cH_r$ can be evaluated in time at most $\Run(\cH)$, and such that $p_2$ is bounded away from $1$. Then there is a fully dynamic algorithm that, on a sequence of $M$ insertions and deletions of points from $\cX$, given an upper bound $M \leq \hat{M} \leq \poly(M)$, with probability $1-\delta$, correctly maintains a $c(2+\eps)$-approximate $k$-centers clustering to the active point set $P^t$ at all time steps $t$, simultaneously for all $k \geq 1$. The total runtime of the algorithm is at most 
\[\tilde{O}\left(M \cdot \frac{\log \Delta \log \delta^{-1}}{\eps p_1} n^{2 \rho} \cdot \Run(\cH) \right)\]
where $\rho = \frac{\ln p_1}{\ln p_2}$, and $n$ is an upper bound on the maximum number of points at any time step.
\end{theorem}
\begin{proof}

We first demonstrate that there exists an algorithm $\cA(\delta,M)$ which takes as input $\delta \in (0,1/2)$ and $M \geq 1$, and on a sequence of at most $M$ insertions and deletions of points in $\cX$, with probability $1-\delta$ correctly solves the top-$M$ LFMIS with leaders (Definition \ref{def:LFMISLead}) on a graph $G$ that is $(r,c,0)$-approximate threshold graph for $P$ (Definition \ref{def:approxThreshGraph}), where $P \subset \cX$ is the set of all points which were inserted during the sequence, and runs in total expected time $\alpha_{\delta,M}$, where 

\[\alpha_{\delta,M} = \tilde{O}\left(M^{1+2\rho} \log(M/\delta)  \Run(\cH)  \right)\]
The reduction from having such an algorithm to obtaining a $k$-centers solution for every $k \geq 1$, incurring a blow-up of $\eps^{-1} \log \Delta$, and requiring one to scale down $\delta$ by a factor of $O(\eps^{-1} \log \Delta)$ so that all $\eps^{-1} \log \Delta$ instances are correct, is the same as in Section \ref{sec:kCenters}, with the modification that the clustering obtained by a MIS $\cL$ at scale $r \geq 0$ has cost at most $cr$, rather than $r$. Thus, in what follows, we focus on a fixed $r$. 

First, setting $s_{\delta,M} = O(\log (M/\delta ) M^{2\rho} / p_1)$ and $t_M = O(\log_{1/p_2} M)$, by Proposition \ref{prop:isAThreshold} it holds that with probability $1-\delta$ the graph $G_{r,cr}(P,\cH,t_M,s_{\delta,M})$ is a $(r,c,L)$-approximate threshold graph, with $\ex{L} < 2$. Then by Theorem \ref{thm:LSHPreMain}, there is an algorithm which maintains a top-$M$ LFMIS with leaders to the graph $G = G_{r,cr}(P,\cH,t_M,s_{\delta,M}) \cap G_{rc}$, which in particular is a $(r,c,0)$-approximate threshold graph for $P$, and runs in expected time at most  

\[O( M ( s_{\delta,M} L \log M+ t_M s_{\delta,M} \Run(\cH)) \log M+ M \log^2 M )\]
where the expectation is taken over the choice of the random ranking $\pi$. Taking expectation over $L$, which depends only on the hash functions $h_1,\dots,h_{s_{\delta,M}}$, and is therefore independent of $\pi$, the expected total runtime is at most $\alpha_{\delta,M}$ as needed. 

To go from the updated time holding in expectation to holding with probability $1-\delta$, we follow the same proof of Proposition \ref{prop:highProb}, except that we set the failure probability of each instance to be $O(\delta/\log^2(M/\delta))$. By the proof of Proposition \ref{prop:highProb}, the total number of copies ever run by the algorithm is at most $O(\log(1/\delta) \sum_{i=1}^{\log M} i) = O(\log(1/\delta) \log^2 M)$ with probability at most $1-\delta/2$, and thus with probability at least $1-\delta$ it holds that both at most $O(\log(1/\delta) \log^2 M)$  copies of the algorithm are run, and each of them is correct at all times. Note that whenever the runtime of the algorithm exceeds this bound, the algorithm can safely terminate, as the probability that this occurs is at most $\delta$ by Proposition \ref{prop:highProb}.

Put together, the above demonstrates the existence of an algorithm $\bar{\cA}(\delta,M)$ which takes as input $\delta \in (0,1/2)$ and $M \geq 1$, and on a sequence of at most $M$ insertions and deletions of points in $\cX$, with probability $1-\delta$ correctly maintains a $c(2+\eps)$-approximation to the optimal $k$-centers clustering simultaneously for all $k \geq 1$, with total runtime at most $\tilde{O}(\eps^{-1} \log \Delta \alpha_{\delta,M})$. However, we would like to only run instances of $\bar{\cA}(\delta,M)$ with $M = O(n)$ at any given point in time, so that the amortized update time has a factor of $n^{2 \rho}$ instead of $M^{2\rho}$. To accomplish this, we greedily pick time steps $1 \leq t_1 < t_2 < \dots < M$ with the property that $t_i - t_{i-1} = \lceil |P^{t_{i-1}}|/2 \rceil$. Observe that for such time steps, we have $|P^{t_{i}}| <  2|P^{t_{i-1}}|$. We then define time steps $1 \leq t_1' < t_2' < \dots < M$ with the property that $t_i'$ is the first time step where the active point set size exceeds $2^i$. We then run an instance of $\bar{\cA}(\delta_0,2^i)$, starting with $i=1$, where $\delta_0 = \delta/M$. Whenever we reach the next time step $t_{i+1}'$ we restart the algorithm $\bar{\cA}$ with parameters $(\delta_0, 2^{i+1})$, except that instead of running $\bar{\cA}$ on the entire prefix of the stream up to time step $t_{i+1}'$, we only insert the points in $P^{t_{i+1}'}$ which are active at that time step. Similarly, when we reach a time step $t_{j}$, we restart $\bar{\cA}$ with the same parameters, and begin by isnerting the active point set $P^{t_{j}}$, before continuing with the stream. 

Note that $\bar{\cA}$ is only restarted $\log n$ times due to the active point set size doubling. Moreover, each time it is restarted due to the time step being equal to $t_j$ for some $j$, the at most $2|P^{t_{j-1}}|$ point insertions required to restart the algorithm can be amortized over the $t_j - t_{j-1} \geq |P^{t_{j-1}}| /2$ prior time steps. Since each instance is correct with probability $1-\delta_0$, by a union bound all instances that are ever restarted are correct with probability $1-\delta$. Note that, since the amortized runtime dependency of the overall algorithm on $M$ is polylogarithmic, substituting $M$ with an upper bound $\hat{M}$ satisfying  $M \leq \hat{M} \leq \poly(M)$ increases the runtime by at most a constant. Moreover, we never run $\bar{\cA}(\delta,t)$ with a value of $t$ larger than $2 n$, which yields the desired runtime.


\end{proof}

The following corollary follows immediately by application of the locally sensitive hash functions of \cite{datar2004locality,har2012approximate}, along with the bounds from Theorem \ref{thm:lshMain}. 

\begin{corollary}\label{cor:Euclidean} Fix any $c \geq 1$. Then there is a fully dynamic algorithm which, on a sequence of $M$ insertions and deletions of points from $d$-dimensional Euclidean space $(\R^d , \ell_p)$, for $p \geq 1$ at most a constant,  with probability $1-\delta$, correctly maintains a $c(4+\eps)$-approximate $k$-centers clustering to the active point set $P^t$ at all time steps $t \in [M]$, and simultaneously for all $k \geq 1$. The total runtime is at most 

\[\tilde{O}\left( M \frac{ \log \delta^{-1} \log \Delta  }{\eps} d n^{1/c}\right)\]
\end{corollary}

For the case of standard Euclidean space ($p=2$), one can used the improved ball carving technique of Andoni and Indyk \cite{andoni2006near} to obtain better locally sensitive hash functions, which result in the following:

\begin{corollary}\label{cor:Euclidean2} Fix any $c \geq 1$. Then there is a fully dynamic algorithm which, on a sequence of $M$ insertions and deletions of points from $d$-dimensional Euclidean space $(\R^d , \ell_2)$,  with probability $1-\delta$, correctly maintains a $c(\sqrt{8}+\eps)$-approximate $k$-centers clustering to the active point set $P^t$ at all time steps $t \in [M]$, and simultaneously for all $k \geq 1$. The total runtime is at most 

\[\tilde{O}\left( M\frac{ \log \delta^{-1} \log \Delta  }{\eps} d n^{1/c^2 + o(1)}\right)\]
\end{corollary}

Additionally, one can use the well-known MinHash \cite{broder1997resemblance} to obtain a fully dynamic $k$-centers algorithm for the \textit{Jaccard Distance}. Here, the metric space is the set of all subsets of a finite universe $X$, equipped with the distance $d(A,B) = 1- \frac{|A \cap B|}{|A \cup B|}$ for $A,B \subseteq X$. We begin stating a standard bound on the value of $\rho$ for the MinHash LSH family.

\begin{proposition}[\cite{indyk1998approximate}]
Let $\cH$ be the hash family given by 
\[ \cH = \{h_\pi :2^X \to X \; | \;  h_\pi(A) = \arg \min_{a \in A} \pi(a), \; \pi \text{ is a permutation of } X \} \] 
Then for any $c \geq 1$ and $r \in [0,1/(2c)]$, we have that $\cH$ is $(r,cr,1-r,1-cr)$-sensitive for the Jaccard Metric over $X$, where $\rho = 1/c$. 
\end{proposition}
\begin{proof}
If $d(A,B) = r$, for any $r \in [0,1]$, we have 
\[ \prb{h\sim \cH}{h(A) = h(B)} = \frac{|A \cap B|}{|A \cup B|} = 1-r \]
Thus, we have $p_1 = 1-r$ and $p_2 = 1-cr$. Now by Claim $3.11$ of \cite{har2012approximate}, we have that for all $x \in [0,1)$ and $c \geq 1$ such that $1-cx > 0$, the following inequality holds:
\[ \frac{\ln(1-x)}{\ln(1-cx)} \leq \frac{1}{c} \]
Thus we have the desired bound:

\[       \rho \leq \frac{\ln(1-r)}{\ln(1- cr)} \leq 1/c\]

\end{proof}

\begin{corollary}\label{cor:Jaccard}
Let $X$ be a finite set, and fix any $c \geq 1$. Then there is a fully dynamic algorithm which, on a sequence of $M$ insertions and deletions of subsets of $X$ equipped with the Jaccard Metric,  with probability $1-\delta$, correctly maintains a $c(4+\eps)$-approximate $k$-centers clustering to the active point set $P^t$ at all time steps $t \in [M]$, and simultaneously for all $k \geq 1$. The total runtime is at most 
\[\tilde{O}\left(M \frac{ \log \delta^{-1} \log \Delta  }{\eps} |X| n^{1/c}\right)\]
\end{corollary}
\begin{proof}
Letting $r_{\min}$ be the minimum distance between points in the stream, we run a copy of the top-$M$ LFMIS algorithm of Theorem \ref{thm:LSHPreMain} for $r = r_{\min}, (1+\eps)r_{\min},\dots,1/(2c)$, where we set the value of the approximation factor $c$ in Theorem \ref{thm:LSHPreMain} to be scaled by a factor of $2$, so that each instance computes a LFMIS on a $(r,2c,0)$-approximate threshold graph. Note that for any time step $t$, if the copy of the algorithm for $r=1/c$ contains a LFMIS with at least $k+1$ vertices, it follows that the cost of the optimal clustering is at most $1/(4c)$. In this case, we can return an arbitrary vertex as a $1$-clustering of the entire dataset, which will have cost at most $1$, as the Jaccard metric is bounded by $1$, and thereby yielding a $4c$ approximation. Otherwise, the solution is at most a 
$c(4+2\eps)$-approximation, which is the desired result after a re-scaling of $\eps$. Note that $\Run(\cH) = \tilde{O}(|X|)$ to evaluate the MinHash, which completes the proof.  

\end{proof}

\section{Deterministic Fully Dynamic $k$-Centers}\label{sec:deterministic}

In this section, we demonstrate the existence of a deterministic algorithm for fully dynamic $k$-centers clustering in arbitrary metric spaces which achieves a $(2+\eps)\lceil \log (n(1+\eps)) \rceil$ approximation in $O(k \log \Delta \log n \log k)$-amortized update time, where $n$ is the maximum number of points which are active at any point in the stream, and $\eps > 0$ is any constant.\footnote{We remark that $\eps$ need not be a constant. In particular, it is easy to see in the proof of Theorem \ref{thm:Deterministic} that, by replacing each of the two occurrences of $\eps$ in the runtime with separate parameters $\eps_0,\eps_1$, the runtime of our algorithm would then depend linearly on $1/(\eps_0 \eps_1)$. We choose to omit this dependency since the approximation is already logarithmic in $n$.} Additionally, our algorithm will have the property that its worst-case update time is at most $O(k^2 \log \Delta \log k \log n)$. Note that the algorithm from Section \ref{sec:generalMetric}, as well as the algorithm of the prior work \cite{chan2018fully}, both run in $\Omega(n)$ time in the worst-case. In what follows, we fix an arbitrary metric space $(\cX, d)$, and any scale $r \in (r_{\min},r_{\max})$. We first describe an algorithm which attempts to generate a $k$-centers clustering with radius at most $r \cdot \lceil \log m\rceil$, where $m$ is the number of updates in the stream (we will later show how $m$ can be replaced with $(1+\eps)n$).

\paragraph{The Algorithm.}  We now describe the data structured maintained by the algorithm for a given scale $r$.
The data structure maintains a binary tree $T$ of depth at most $\lceil \log m \rceil$. Each node $v$ in the tree will contain a subset $S_v$ of the active input points, where each leaf contains a unique individual point. Moreover, we will maintain a subset $U_c \subseteq S_v$ of ``uncovered'' points. For any vertex $v$ with right and left children $v_r,v_l$, we will have $S_v \subset U_{v_r} \cup U_{v_l}$. In other words, only the uncovered points have an opportunity to move up to the next level of the tree. We will maintain the property that $|U_v| \leq |S_v| \leq 2k+2$ for all nodes $v \in T$. 

When $|U_v|$ becomes larger than $k+1$, we will not insert additional points into $S_v$ until $|U_v|$ drops below $k+1$. To accomplish this, each vertex $v$ will maintain a queue $Q_v$ of points. Whenever $|U_v| \geq k+1$, we will add new points propagating up to $v$ into $Q_v$, instead of adding them directly to $S_v$. Additionally, if $w$ is either the left $v_\ell$ or right $v_r$ child of $v$, we will not propagate more than $k+1$ points up from $w$ to $v$. Thus, if $|U_w| \geq k+1$, additional points beyond the first $k+1$ will be added to the queue $Q_v$ instead of directly to $S_v$. We only will begin inserting points from $Q_v$ into $S_v$ when all three of $U_v,U_{v_\ell},U_{v_r}$ have size less than $k+1$. For any vertex $v$, we call $v$ \textit{halted} if one of $U_v,U_{v_\ell},U_{v_r}$ has size at least than $k+1$. 

The final clustering of the algorithm will be given by the centers in $U_v$, where $v$ is the root of $T$. Note that points in a queue $Q_u$ for any $u \in T$ are not contained in any formal cluster, since they are not necessarily covered by any point in $u$. On the other hand, if any queue $Q_u$ is non-empty, by the above this implies that there exists at least $k+1$ points in the active set which are pairwise distance at least $r$ apart, which serves as a certificate that the cost of the optimal solution is at least $r/2$. Thus, for such a scale $r$, we will not be using the clustering at that scale as our output at that time step anyway. 

Formally, our data structure will satisfy the following properties. 

\begin{figure}[H]
\begin{Frame}[Properties of data structures associated with every $v \in T$]
\begin{enumerate}
\item Once a point $p$ is added to the set $S_v$ (resp. $U_v$), it remains in $S_v$ (resp. $U_v$) until it is deleted from the stream. 
\item We have $|U_v| \leq |S_v| \leq 2k+2$.
    \item For all distinct uncovered points $x,y \in U_v$, we have $d(x,y) > r$
    \item For all  $x \in S_v \setminus U_v$, there is at least one point $y \in U_v$ with $d(x,y) \leq r$. 
    \item If $|Q_v| > 0$, then we have that at least one of $|U_v| \geq k+1$, $|U_{v_\ell}| \geq k+1$, or $|U_{v_r}| \geq k+1$  holds, where $v_\ell,v_r$ are the left and right children, respectively, of $v$ in $T$. 
\end{enumerate}
\end{Frame}
\caption{Invariants of the data structure}\label{fig:invar}
\end{figure}
Each point $x \in S_v$ is given an unique identifier in the set $[2k+2]$ when it is first inserted into $S_v$. If $x$ is later deleted from $S_v$, its identifier is freed up to be potentially used by another point which is inserted into $S_v$ later on; the set of available identifiers can be stored in a queue. 
For each $v \in T$ and point $x \in S_v$, we will maintain the set of points in $U_v$ and $S_v\setminus U_v$ which are neighbors of $x$ in the $r$-threshold graph on the point set $S_v \subset \cX$. Specifically, if we define $N_r(x) = \{z \in \cX \; | \; d(x,z) \leq r\}$, then within the vertex $v$ we dynamically maintain the two sets $N_r(x) \cap U_v$ and $N_r(x) \cap (S_v\setminus U_v)$. These two sets are maintained in a binary search tree, using the $O(\log k)$-bit identifiers for the points in $S_v$, so that insertions, deletions, and searching for an arbitrary member of the set can each be deterministically accomplished in $O(\log k)$ time. We call $N_r(x) \cap U_v$ the set of \textit{uncovered neighbors of $x$} at $v$, and we call $N_r(x) \cap (S_v\setminus U_v)$  the set of \textit{covered neighbors of $x$} at $v$.

\paragraph{The Binary-Tree Data Structure}
We now formally describe how we maintain the binary tree data structure $T$. We will think of $T$ as being a subtree of an infinite binary tree $\overline{T}$, with leaf nodes labelled $\{1,2,\dots\}$, where the root of $T$ is the rightmost node at the level $i$ such that $2^i$ is larger than the current number of points which have been inserted into the stream (where we think of leaf nodes as being at level $1$). On the $j$-th insertion of a point point $p$ in the stream, we add the $j$-th leaf node of $\overline{T}$ to the tree $T$, possibly resulting in the root of $T$ being shifted one level up $\overline{T}$, and we additionally add any nodes to $T$ which lie on the root-to-leaf path between the $j$-th leaf and new root of $T$; note that updating the root and constructing such nodes requires at most $O(\log m)$ time.  We then add the information for the point $p$ in the $j$-th leaf node, and begin propagating  the point $p$ up the path to the root (the method for propagating a point up a path is described below). When a point $q \in \cX$ is deleted, we simply delete $q$ from the leaf $\ell$ corresponding to $q$, and then propagate the removal of $q$ up the tree from $\ell$ by deleting $q$ from each set $S_v$ such that $q \in S_v$. 

In what follows, we now describe how insertions and deletions of points are propagated in the tree.

\paragraph{Insertions.}
Once a new point $p \in \cX$ is added to the stream, and we create a new rightmost leaf-node $\ell$ and add $p$ to that node, we will then attempt to insert $p$ into the sets $S_v$ for each $v$ on the path from $\ell$ to the root of $T$, until either we reach the root of $T$, we arrive at a vertex $v$ such that $N_r(p) \cap U_v$ is non-empty (in which case, $p$ is covered by a point in $U_v$), or we arrive at a vertex $v$ with $|U_v| \geq k+1$ (in which case we will add $p$ to the queue $Q_v$). 

Formally, starting with $v = \ell$, we perform the following steps. First, if $|U_v| \geq k+1$, we add $p$ to $Q_v$ and terminate; otherwise, we add $p$ to $S_v$, assign it a unique identifier in $[2k+2]$ for $S_v$, and compute the two neighborhood sets $N_r(p) \cap U_v$ and $N_r(p) \cap (S_v\setminus U_v)$, storing the results in two binary search trees as described above. If $N_r(p) \cap U_v = \emptyset$, we add $p$ to $U_v$. Additionally, for each $q \in N_r(p) \cap S_v$, we add $p$ to the relevant neighborhood list of $q$, namely either the uncovered neighbors of $q$ if we added $p$ to $U_v$, or the covered neighbors of $q$ otherwise.  If $N_r(p) \cap U_v \neq \emptyset$, we terminate the procedure here. Then if $v$ is the root, we terminate the procedure, otherwise we update $v \leftarrow \ttx{parent}_T(v)$ and continue up the tree. The procedure to insert a point $p$ is described below in Algorithm \ref{alg:insDet}, where the algorithm Propoagate is called beginning with the leaf $v = \ell$ corresponding to the point $p$. 

\begin{algorithm}[H]
\DontPrintSemicolon
	\caption{Propagate$(p,v)$} \label{alg:insDet}
\While{processing of $p$ not yet terminated}{

\If{$|U_v| \geq k+1$}{
    Add $p$ to $Q_v$. \;
    Terminate processing of $p$. 
}\Else{
Add $p$ to $S_v$, and assign $p$ an identifier within $v$ from the set $[2k+1]$. \;
Compute the sets $N_r(p) \cap U_v$ and $N_r(p) \cap (S_v\setminus U_v)$ if they were not already computed. \;
\If{$N_r(p) \cap U_v = \emptyset$}{
    Add $p$ to $U_v$. \;
    For each $q \in N_r(p) \cap S_v$, add $p$ to $N_r(q) \cap U_v$ (the uncovered neighborhood list of $q$). \;
}
\Else{
For each $q \in N_r(p) \cap S_v$, add $p$ to $N_r(q) \cap (S_v \setminus U_v)$ (the covered neighborhood list of $q$). \;
 Terminate processing of $p$. 
}

}
\If{$v$ is the root of $T$}{
    Terminate processing of $p$.  \; 
    }
    \Else{
$v \leftarrow  \ttx{parent}_T(v)$. \;
}

}
\end{algorithm}

\paragraph{Deletions.} We now describe how to delete a point $q \in \cX$. Let $\ell$ be the leaf node in $T$ corresponding to $q$. Firstly, if $q \in Q_v$ for some vertex $v$, we delete $q$ from $Q_v$. Then, for each vertex $v$ with $q \in S_v$, beginning with the $v$ highest in the tree, we perform the following. First, if $q \notin U_v$, we simply remove $q$ from $S_v$, and delete $q$ from the relevant neighborhood lists of all points in $N_r(q) \cap S_v$. If $q \in U_v$, we also remove $q$ from $S_v$, and delete $q$ from the  neighborhood lists of all points in $N_r(q) \cap S_v$; however, in this case, after removing $q$, we then attempt to reinsert each point $p \in N_r(q) \cap S_v$ by propagating it up the tree as follows. Firstly, for any such  $p \in N_r(q) \cap S_v$, if $N_r(p) \cap U_v \neq \emptyset$ even after removing $q$, we do nothing for $p$ as it is still covered by some point in $U_v$. Otherwise, we add $p$ to $U_v$, move $p$ from the set of covered neighbors to the set of uncovered neighbors of each $q \in N_r(p) \cap S_v$. We then attempt to insert $p$  attempt to insert $p$ into the parent $u$ of $v$ as described in the insertion paragraph above, however, if $|U_v| \geq k+1$ (as the result of several vertices becoming uncovered after the removal of $q$), or if $|U_u| \geq k+1$, then we instead add $p$ to the queue $Q_u$ at the parent. 

Once all such newly uncovered points in the vertex $v$ have been processed, we move to the queue $Q_v$. Specifically, $|Q_v| > 0$, and if Property $4$ in Figure \ref{fig:invar} no longer holds (namely, if $v$ and both of the children of $v$ have fewer than $k+1$ uncovered points), we then add the points from $Q_v$ into $S_v$, following the steps above for inserting a point into $S_v$, until either $Q_v$ is empty or the number of uncovered points at $v$ exceeds $k$. After this step is complete, we move on to the next vertex $v'$ (either the left or right child of $v$) on the path down to the leaf $\ell$, and delete $q$ from $S_{v'}$ in the same fashion as just described. The algorithm to delete a point $q$ is given concretely in Algorithm \ref{alg:delDet}. 

\begin{algorithm}[H]
\DontPrintSemicolon
	\caption{Delete$(q)$} \label{alg:delDet}
	$\ell :=$ leaf corresponding to $q$. \;
	Let $\ell = v_1 ,v_2,\dots,v_t = \ttx{root}(T)$ be the ordered vertices on the leaf-to-root path between $\ell$ and the root of $T$. \;
	 Set $i$ to be the largest $i \in [t]$ such that $q \in S_{v_i}$. \; 
	 \If{$q \in Q_{v_{i+1}}$}{
	 Delete $q$ from $Q_{v_{i+1}}$. \;
	 }
	\While{$i \geq 1$}{
	Delete $q$ from $S_{v_i}$.\;
	 \If{$q \in U_{v_{i}}$}{
	 Delete $q$ from $U_{v_i}$. \;
	 \For{each $p \in N_r(q) \cap S_{v_i}$}{
	 Update $N_r(p) \cap U_{v_i}$ (the uncovered neighborhood list of $p$) by removing $q$. \;
	 \If{$N_r(p) \cap U_{v_i} = \emptyset$}{
	 Add $p$ to $U_{v_i}$. \; 
	 For each $z \in N_r(p) \cap S_{v_i}$, move $p$ from the set of covered neighbors to the set of uncovered neighbors of $z$ at $v_i$. \;
	  \If{$|U_{v_i}| \leq k+1$}{
	 Call Propagate$(p,v_{i+1})$. \; 
	 }\Else{
	 Add $p$ to $Q_{v_{i+1}}$. \;
	 }
	 }
	 }
	}\Else{
	 For each $p \in N_r(q) \cap S_{v_i}$, update the set $N_r(p) \cap (S_{v_i} \setminus U_{v_i})$ (the covered neighborhood list of $p$) by removing $q$. \;
	}
	$i \leftarrow i-1$. \; 
	}
	$i \leftarrow 1$
		\While{$i \leq t$}{
		\While{$v_i$ is not halted and $|Q_{v_i}| > 0$}{
		Remove the first point $z$ from the queue $Q_{v_i}$, and call Propagate$(z,v_i)$. 
		}
		$i \leftarrow i+1$. \;
		}
	
	\end{algorithm}

\begin{proposition}\label{prop:det1}
After processing each update in the stream, all properties in Figure \ref{fig:invar} hold. 
\end{proposition}
\begin{proof}
For the first property, this follows directly from the fact that the only time a point $q$ is every removed from either $S_v$ or $U_v$ is when Delete$(q)$ is called. For the second property, since points are only added to $S_v$ for a non-leaf node $v$ when they are either contained in either $U_{v_\ell}$ or $U_{v_r}$, it suffices to show that $|S_v \cap U_w| \leq k+1$ where $w$ is either the left or right child of $v$. Note that the size of $U_w$ can only exceed $k+1$ when we delete a point $q$ from $U_w$, causing more than one point in $S_w$ which was previously covered by $q$ to be added to $U_w$. In this case, note that we only call Propogate$(p,v)$ for such points $p$ (inside of Algorithm \ref{alg:delDet}) when $|U_w| \leq k+1$, otherwise we add $p$ to the queue $Q_v$, along with all excess points added to $U_w$ beyond the first $k+1$ on that step. Thereafter, no further points can be added to $S_w$ until $|U_w| < k+1$. Therefore, at each step at most $k+1$ points in $U_w$ are contained in $S_v$, which demonstrates the second property. 

The third property holds because we only add a point $p$ to $U_v$ if is distance at least $r$ from all points which were at $U_v$ at that time. The fourth property holds because, for each point $p \in S_v$, we only fail to add it to $U_v$ if $N_r(p) \cap U_v \neq \emptyset$. The final point holds because a point is only added to $Q_v$ when $v$ is halted, and $v$ can only cease to be halted when a point $q$ is deleted from $v$, in which case all points from $Q_v$ are removed from $Q_v$ and inserted into $S_v$ (at the end of Algorithm \ref{alg:delDet}) until either $Q_v$ is empty or $v$ is again halted. 

for any vertex $v$ note that a point $p$ can only be inserted into $S_v$ during a call to the Propagate$(p, \cdot)$ procedure. WLOG the leaf corresponding to $\ell$ is contained in the left subtree of $v$. Then 

\end{proof}

\begin{proposition}\label{prop:det2}
The total update time of the algorithm, on a sequence of $m$ insertions and deletions, is at most $O(km \log k \log m)$. Moreover, the worst-case runtime of any insertion is at most $O(k \log k \log m)$, and the worst case update time of any deletion is $O(k^2 \log k \log m)$. . 
\end{proposition}
\begin{proof}
For each point $p$ which arrives in the stream, we give $p$ $O(k \log k)$ credits for each of the $O(\log m)$ levels in the tree. For each level $i$, let $v$ be the node at level $i$ which contains $p$ in its subtree. Note by the prior proposition that $p$ is added
to each of $S_v$ and $U_v$ exactly once. When it is first added to $S_v$, we require $O(k)$ time to search through all points in $S_v$ and check if they are contained in $N_r(p)$. Since $|S_v| = O(k)$ at all times, this requires $O(k)$ distance computations. Then for each such point in $S_v \cap N_r(p)$, we add them to a binary search tree storing the neighborhoods for $p$, which requires $O(k \log k)$ total time. If $p$ is also added to $U_v$ at the same time it is added to $S_v$, no further work is charged to $p$ at this time. Otherwise, $p$ may later be added to $U_v$ during a Delete$(q)$ call because $p$ becomes uncovered at this point. We then charge $p$ another $O(k \log n)$, which is the runtime cost of moving $p$ from the covered neighborhood set to the uncovered neighborhood set of each of its neighbors in $S_v$. Moreover, note that each point is added to a queue at level $i$ at most once.   This covers the cost of propagating each point to every level in the tree. 

When we delete a point $q$, we consider at most $O(k)$ neighbors in each set $N_r(q) \cap S_v$ (for each vertex $v$ with $q \in S_v$). For each such neighbor $p$, if they do not become uncovered at that point, we do constant work to verify that $N_r(p) \cap U_v$ is still non-empty. Otherwise, if $p$ does become uncovered, it enters $U_v$ for the first time, and the work required to update the neighborhood sets of neighbors of $p$ is charged to the point $p$ as described above. Thus, a total of $O(k \log m)$ is charged to $q$ for a deletion of $q$. This completes the proof of the total runtime.

To see the worst-case runtime, note that the worst case runtime of an insertion is at most $O(k \log m \log k)$, coming from propagating a point $p$ up all levels of the tree. For deletions, note that deleting $q$ causes at most $O(k)$ points to become uncovered at each level, and for each level at most $k+1$ newly uncovered points can be propagated up to the next level on that step. Thus, at most $O(k \log m)$ total calls to Propagate are made for freshly uncovered points, which requires a total runtime of $O(k^2 \log m \log k)$. Finally, for the runtime at the end of Algorithm \ref{alg:delDet} required to process queues $Q_v$ of vertices $v$ which became no longer halted after the deletion, we note that  $|Q_v| = O(k)$ follows naturally because $Q_v \subset U_{v_\ell} \cup U_{v_r}$, therefore at most $k$ vertices are propagated out of each queue within a single call to Delete $(q)$, each of which runs in $O(k \log m \log k)$ time, which completes the proof. 
\end{proof}

\begin{theorem}\label{thm:Deterministic} 

There is a deterministic algorithm that, on a sequence of insertions and deletions of points from an arbitrary metric space $(\cX,d)$, maintains a $(2+\eps)\lceil \log (n(1+\eps)) \rceil$-approximation to the optimal $k$-centers clustering, where $\eps>0$ is any constant. The amortized update time of the algorithm is $O(k \log \Delta \log n \log k   )$. The worst-case update time of any insertion is $O(k \log \Delta \log n  \log k)$, and the worst-case update time of any deletion is $O(k^2 \log \Delta \log n  \log k  )$. 
\end{theorem}
\begin{proof}
First, we argue that the Propositions \ref{prop:det1} and \ref{prop:det2} imply the existence of an algorithm that maintains a $(2+\eps)\lceil \log m \rceil$-approximation to the optimal $k$-centers clustering, in amortized update time  $O(\frac{\log \Delta \log m \log k}{\eps}k  )$ and worst-case update time  $O(\frac{\log \Delta \log m}{\eps}k^2  )$, where $m$ is the length of the stream. To see this, we run the above algorithm with $r = (1+\eps)^i r_{\min}$ for each $i=0,1,2,\dots\frac{1}{\eps}\log \Delta$.Any $r$ that contains a halted vertex $v$ gives a certificate that the cost of the optimal solution is at most $r/2$, since a halted vertex only occurs when some set $U_w$ contains at least $k+1$ points which are pairwise distance at least $r$ apart. It suffices to show that the smallest scale $r$ without a halted vertex gives a clustering with cost at most $r \lceil \log m \rceil$. The cluster centers in question are the points in $U_{\ttx{root}(T)}$ where $\ttx{root}(T)$ is the root of $T$ (note there must be at most $k$ such points since $v$ is not halted. Now note that every queue in this data structure must be empty by Property $5$ of Figure \ref{fig:invar}, thus $S_v = U_{v_\ell} \cup U_{v_r}$ for every non-leaf vertex $v$. Thus, we can trace a series of pointers $p= p_1,p_2,\dots,p_{\lceil \log m \rceil}$, where $p_i = p_{i+1}$ if $p_{i}$ is contained in set $U_{v_i}$, where $v_i$ is the vertex at level $i$ up the tree from the leaf containing $p$, otherwise we set $p_i$ to be any point in $U_{v_{i}}$ which is distance at most $r$ away from $p$ (which exists due to Property $r$ of Figure \ref{fig:invar}). In each case, $d(p_i,p_{i+1}) \leq r$, thus $d(p,p_{\lceil \log m \rceil}) \leq \lceil \log m \rceil r$, from which the claim follows because $p_{\lceil \log m \rceil}$ is contained in the clustering $U_{\ttx{root}(T)}$. 

It remains to limit the height of the tree from $\lceil \log m \rceil$ to $\lceil \log n(1+\eps) \rceil $. To do this, we rebuild the tree (and the entire data structure) whenever the number of new updates exceeds a $\eps$ fraction of the number of active points on the last time the tree was rebuilt. In other words, if the last rebuild occurred on on time $t_i$ when there were $n_i$ active points, after the next $\eps n_i$ updates we delete the current tree, and take all active points remaining and insert them into a new tree (which now contains only active points). The runtime can then be amortized over the $\eps n_i$ intermediate updates. So as to not effect the worst-case update time, we use the standard procedure of running the rebuilding procedure in the background (i.e. we buffer the insertions into the new tree). Namely, on each of the $\eps n_i$ intermediate steps, we take $2/\eps$ arbitrary active points in the current set and add them to the new tree. Thus, by the time we have arrived at the $\eps n_i$ step, at most an $\eps$ fraction of the leafs in the newly rebuilt tree correspond to deleted points. Thus, the height of the tree is always bounded by $\lceil \log n(1+\eps) \rceil$, thus the total approximation is at most $(2+\eps)\lceil \log n(1+\eps) \rceil$, which is as desired.
\end{proof}

\section{Lower Bound for Arbitrary Metrics}\label{sec:LB}
We now demonstrate that any algorithm which approximates the optimal $k$-centers cost, in an arbitrary metric space of $n$ points, must run in $\Omega(nk)$ time. Specifically, the input to an algorithm for $k$-centers in arbitrary metric spaces is both the point set $P$  \textit{and} the metric $d$ over the points. In particualr, the input can be represented via the distance matrix  distance matrix $\bD \in \R^{n \times n}$ over the point set $P$, and the behavior of such an algorithm can be described by a sequences of adaptive queries to $\bD$. 

The above setting casts the problem of approximating the cost of the optimal $k$-centers clustering as a \textit{property testing} problem  \cite{goldreich1998property,goldreich2017introduction}, where the goal is to solve the approximation problem while making a small number of queries to $\cD$. Naturally, the query complexity of such a clustering algorithm lower bounds its runtime, so to prove optimality of our dynamic $k$-centers algorithms it suffices to focus only on the query complexity. In particular, in what follows we will demonstrate that any algorithm that approximates the optimal $k$-centers cost to any non-trivial factor with probability $2/3$ must query at least $\Omega(nk)$ entries of the matrix.  In particular, this rules out any fully dynamic algorithm giving a non-trivial approximation in $o(k)$ amortized update time for general metric spaces. 

Moreover, we demonstrate that this lower bound holds for the $(k,z)$-clustering objective, which includes the well studied $k$-medians and $k$-means. Recall that this problem is defined as outputting a set $\cC \subset \cX$ of size at most $k$ which minimizes the objective function
\[ \cost_{k,z}(\cC) =   \sum_{p \in P} d^z(p,\ell(p))    \]
where $\ell(p)$ is the cluster center associated with the point $p$. 
Note that  $(\cost_{k,z}(\cC))^{1/z}$ is always within a factor of $n$ of the optimal $k$-centers cost. Thus, it follows that if $k$-centers cannot be approximated to any non-trivial factor (including factors which are polynomial in $n$) in $o(nk)$ queries to $\bD$, the same holds true for $(k,z)$-clustering for any constant $z$. Thus, in the proofs of the following results we focus solely on proving a lower bound for approximation $k$-centers to any factor $R$, which will therefore imply the corresponding lower bounds for $(k,z)$-clustering. 

We do so by first proving Theorem \ref{thm:LBBig}, which gives a $\Omega(nk)$ lower bound when $n = \Omega(k \log k)$. Next, in Proposition \ref{prop:lb}, we prove a general $\Omega(k^2)$ lower bound for any $n > k$, which will complete the proof of Theorem \ref{thm:LBMain}. We note that the proof of Proposition \ref{prop:lb} is fairly straightforward, and the main challenge will be to prove Theorem \ref{thm:LBBig}.

\begin{theorem} \label{thm:LBBig}
Fix and $k \geq 1$ and $n > C k \log k$ for a sufficiently large constant $C$. Then any algorithm which, given oracle access the distance matrix $\cD \in \R^n$ of a set $X$ of $n$ points in a metric space, determines correctly with probability $2/3$ whether the optimal $k$-centers cost on $X$ is at most $1$ or at least $R$, for any value $R >1$, must make at least $\Omega(k n)$ queries in expectation to $\cD$. The same bound holds true replacing the $k$-centers objective with $(k,z)$-clustering, for any constant $z>0$. 
\end{theorem}
\begin{proof}
We suppose there exists such an algorithm that makes at most an expected $k n/8000$ queries to $\cD$. By forcing the the algorithm to output an arbitrary guess for $c$ whenever it queries a factor of $20$ more entries than its expectation, by Markov's inequality it follows that there is an algorithm which correctly solves the problem with probability $2/3 - 1/20 > 6/10$, and always makes at most $kn/400$ queries to $\cD$.

\paragraph{The Hard Distribution.}
We define a distribution $\cD$ over $n \times n$ distance matrices $\cD$ as follows. First, we select a random hash function $h:[n] \to [k]$, a uniformly random coordinate $i \sim [n]$. We then set $\bD(h)$ to be the matrix defined by 
$\bD_{p,q}(h) = 1$ for $p \neq q$ if $h(p) = h(q)$, and $\bD_{p,q}(h) = R$ otherwise, where $R$ is an arbitrarily large value which we will later fix. We then flip a coin $c \in \{0,1\}$. If $c=0$, we return the matrix $\bD(h)$, but if $c = 1$, we define the matrix $\bD(h,i)$ to be the matrix resulting from changing every $1$ in the $i$-th row and column of $\bD(h)$ to the value $2R$. It is straightforward to check that the resulting distribution satisfies the triangle inequality, and therefore always results in a valid metric space. We write $\cD_0 ,\cD_1$ to denote the distribution $\cD$ conditioned on $c=0,1$ respectively. In the testing problem, a matrix $\bD \sim \cD$ is drawn, and the algorithm is allowed to make an adaptive sequence of queries to the entries of $\bD$, and thereafter correctly determine with probability $2/3$ the value of $c$ corresponding to the draw of $\bD$. 

Note that a draw from $\cD$ can then be described by the values $(h,i,c)$, where $h \in \cH = \{h': [n] \to [k] \}$, $i\in[n]$, and $c \in \{0,1\}$. Note that, under this view, a single matrix $\bD \sim \cD_0$ can correspond to multiple draws of $(h,i,0)$. Supposing there is a randomized algorithm which is correct with probability $6/10$ over the distribution $\cD$ and its own randomness, it follows that there is a deterministic algorithm $\cA$ which is correct with probability $6/10$ over just $\cD$, and we fix this algorithm now. 

Let $(d_1,p_1),(d_2,p_2),\dots,(d_t,p_t)$ be an adaptive sequence of queries and observations made by an algorithm $\cA$, where $d_i \in \{1,R,2R\}$ is a distance and $p_i \in \binom{n}{2}$ is a position in $\cD$, such that the algorithm queries position $p_i$ and observed the value $d_i$ in that position.

\begin{claim}\label{claim:lb1}
There is an algorithm with optimal query vs. success probability trade-off which reports $c=1$ whenever it sees an entry with value $d_i = 2R$, otherwise it reports $c=0$ d if it never sees a distance of value $2R$. 
\end{claim}
\begin{proof}
 To see this, first note that if $c=0$, one never sees a value of $2R$, so any algorithm which returns $c=0$ after observing a distance of size $2R$ is always incorrect on that instance. 

For the second claim, suppose an algorithm $\cA$ returned that $c=1$ after never having seen a value of $2R$. Fix any such sequence  $S= \{(d_1,p_1),(d_2,p_2),\dots,(d_t,p_t)\}$ of adaptive queries and observations such that $d_i \neq 2R$ for all $i=1,2,\dots,t$. We claim that $\pr{c=0 | S} \geq \pr{c=1 |S}$. To see this, let $(h,i,1)$ be any realization of a draw from $\cD_1$, and note that $\pr{(h,i,1)} = \pr{(h,i,0)} = \frac{1}{2 n } k^{-n}$. Let $F_0(S)$ be the set of tuples $(h,i)$ such that the draw $(h,i,0)$ could have resulted in $S$, and $F_1(S)$ the set of tuples $(h, i)$ such that $(h,i,1)$  could have resulted in $s$. Let $(h,i,1)$ be a draw that resulted in $S$. Then $(h,i,0)$ also results in $S$, because the difference between the resulting matrices $\cD$ is supported only on positions which were initially $2R$ in the matrix generated by $(h,i,1)$. Thus $F_1(S) \subseteq F_2(S)$, which demonstrates that $\pr{c=0 | S} \geq \pr{c=1 |S}$. Thus the algorithm can only improve its chances at success by reporting $c=0$, which completes the proof of the claim

\end{proof}

\paragraph{Decision Tree of the Algorithm.}
The adaptive algorithm $\cA$ can be defined by a $3$-ary decision tree $T$ of depth at most $ k n/400$, where each non-leaf node $v \in T$ is labelled with a position $p(v) = (x_v,y_v) \in [n] \times [n]$, and has three children corresponding to the three possible observations $\bD_{p(x)} \in \{1,R,2R\}$. Each leaf node contains only a decision of whether to output $c=0$ or $c=1$. For any $v \in T$, let $v_1,v_{R}, v_{2R}$ denote the three children of $v$ corresponding to the edges labelled $1,R$ and $2R$,
Every child coming from a ``$2R$'' edge is a leaf, since by the above claim the algorithm can be assumed to terminate and report that $c=1$ whenever it sees the value of $2R$. For any vertex $v \in T$ at depth $\ell$, let $S(v) = \{ (d_1,p_1),\dots,(\cdot, p(v))\}$ be the unique sequence of queries and observations which correspond to the path from the root to $v$. Note that the last entry $(\cdot,  p(v)) \in S(v)$ has a blank observation field, meaning that at $v$ the observation $p(v)$ has not yet been made.

For any $v \in T$ and $i \in [n]$, we say that a point $i$ is \textit{light} at $v$ if the number of queries $(d_j,p_j) \in S $ with $i \in p_j$ is less than $k/2$. If $i$ is not light at $v$ we say that it is \textit{heavy} at $v$. For any $i \in [n]$, if in the sequence of observations leading to $v$ the algorithm observed a $1$ in the $i$-th row or column, we say that $i$ is \textit{dead} at $v$, otherwise we say that $i$ is \textit{alive}. We write $\pr{v}$ to denote the probability, over the draw of $\bD \sim \cD$, that the algorithm traverses the decision tree to $v$, and $\pr{v | \; c= b}$ for $b \in \{0,1\}$ to denote this probability conditioned on $\bD \sim \cD_b$.  Next, define $F_b(v) = F_b(S(v))$ for $b \in \{0,1\}$, where $F_b(S)$ is as above. 
Note that if $(h,i) \in F_0(v)$ for some $i \in [n]$, then $(h,j) \in F_0(v)$ for all $j \in [n]$, since the matrices generated by $(h,i,0)$ are the same for all $i \in [n]$. Thus, we can write $h \in F_0(v)$ to denote that $(h,i) \in F_0(v)$ for all $i \in [n]$.


\begin{claim}\label{claim:lb2}
Let $v \in T$ be a non-leaf node where at least one index $i$ in $p(v) = (i,j)$ is alive and light at $v$. Then we have \[\prb{\bD \sim \cD}{\bD_{p(v)} = 1 | S(v), c=0} \leq \frac{2}{k}\]
\end{claim}
\begin{proof}
Fix any function $h \in \cH$ such that $h \in F_0(v)$: namely, $h$ is consistent with the observations seen thus far. 
Let $h_1,\dots,h_k \in \cH$ be defined via $h_t(j) = h(j)$ for $j \neq i$, and $h_t(i) = t$, for each $t \in [k]$. We claim that $h_t \in F_0(v)$ for at least $k/2$ values of $t$. 
To show this, first note that the values of $\{h(j)\}_{j \neq i}$ define a graph with at most $k$ connected components, each of which is a clique on the set of values $j \in [n] \setminus \{i\}$ which map to the same hash bucket under $h$. 
The only way for $h_t \notin F_0(v)$ to occur is if an observation $(i,\ell)$ was made in $S(v)$ such that $h(\ell) = t$. Note that such an observation must have resulted in the value $R$, since $i$ is still alive (so it could not have been $1$). In this case, one knows that $i$ was not in the connected component containing $\ell$. However, since $i$ is light, there have been at most $k/2$ observations involving $i$ in $S(v)$. Each of these observations eliminate at most one of the $h_t$'s, from which the claim follows. 

Given the above, it follows that if at the vertex $v$ we observe $\bD_{p(v)} = \bD_{i,j} = 1$, then we eliminate every $h_t$ with $t \neq h(j)$ and $h_t \in F_0(v)$. Since for every set of values $\{h(j)\}_{j \neq i}$ which are consistent with $S(v)$ there were $k/2$ such functions $h_t \in F_0(v)$, it follows that only a $2/k$ fraction of all $h \in \cH$ result in the observation $\bD_{p(v)} = 1$. Thus, $|\cF_0(v_1)| \leq \frac{2}{k}|\cF_0(v)|$, which completes the proof of the proposition.
\end{proof}

Let $\cE_1$ be the set of leafs $v$ which are children of a $2R$ labelled edge, and let $\cE_0$ be the set of all other leaves. Note that we have $\pr{v \; | \; c=0} = 0$ for all $v \in \cE_1$, and moreover $\sum_{v \in \cE_0} \pr{v | c=0} = 1$. For $v \in T$, let $\theta(v)$ denote the number of times, on the path from the root to $v$, an edge $(u,u_1)$ was crossed where at least one index $i \in p(u)$ was alive and light at $u$. Note that such an edge kills $i$, thus we have $\theta(v) \leq n$ for all nodes $v$. Further, define $\hat{\cE}_0 \subset \cE_0$ to be the subset of vertices $v \in \cE_0$ with $\theta(v) < n/20$. We now prove two claims regarding the probabilities of arriving at a leaf $v \in \hat{\cE}_0$. 

\begin{claim}\label{claim:lb3}
Define $\hat{\cE}_0$ as above. Then we have 
\[  \sum_{v \in \hat{\cE}_0} \pr{v | c=0} > 9/10 \]
\end{claim}
\begin{proof}
We define indicator random variables $\bX_1,\bX_2,\dots,\bX_t \in \{0,1\}$, where $t \leq  kn/400$ is the depth of $T$, such that $\bX_i = 1$ if the $i$-th observation made by the algorithm causes a coordinate $i \in [n]$, which was prior to observation $i$ both alive and light, to die, where the randomness is taken over a uniform draw of $\bD \sim \cD_0$. Note that the algorithm may terminate on the $t'$-th step for some $t'$ before the $t$-th observation, in which case we all trailing variables $\bX_{t'},\dots,\bX_t$ to $0$. By Claim \ref{claim:lb2}, we have $\ex{\bX_i} \leq 2k$ for all $i \in [t]$, so $\ex{\sum_{i \in [t]} \bX_i} <  n/ 200$. By Markov's inequality, we have $\sum_{i \in [t]} \bX_i < n/20$ with probability at least $9/10$. Thus with probability at least $9/10$ over the draw of $\bD \sim \cD_0$ we land in a leaf vertex $v$ with $\theta(v) < n/20$, implying that $v \in \hat{\cE}_0$ as needed.

\end{proof}

\begin{claim}\label{claim:lb4}
For any $v \in \hat{\cE_0}$, we have $  \pr{v \; | \; c=1} > (9/10 )\pr{v \; | \; c=0}$.
\end{claim}
\begin{proof}
Fix any $v \in \hat{\cE}_0$. By definition, when the algorithm concludes at $v$, at most $n/5$ indices were killed while having originally been light. Furthermore, since each heavy index requires by definition at least $k/2$ queries to become heavy, and since each query contributes to the heaviness of at most $2$ indices, it follows that at most $kn/400 (4/k) = n/100$ indices could ever have become heavy during any execution. Thus there are at least $n - n/20 - n/100 > (9/10)n$ indices $i$ which are both alive and light at $v$. 

Now fix any $h \in \cF_0(v)$. We show that $(h,i) \in \cF_1(v)$ for at least $(9/10)n$ indices $i \in [n]$, which will demonstrate that $|\cF_1(v)| > (9/10)  |\cF_0(v)|$, and thereby complete the proof. In particular, it will suffice to show that is true for any $i \in [n]$ which is alive at $v$. To see why this is the case, note that by definition if $i$ is alive at $v \in \cE_0$, then $S(v)$ includes only observations in the $i$-th row and column of $\cD$ which are equal to $R$. It follows that none of these observations would change if the input was instead specified by $(h,i,1)$ instead of $(h,j,0)$ for any $j \in [n]$, as the difference between the two resulting matrices are supported on values where $\cD$ is equal to $1$ in the $i$-th row and column of $\cD$. Thus if $i$ is alive at $v$, we have that $(h,i) \in \cF_1(v)$, which completes the proof of the claim.
\end{proof}
\noindent Putting together the bounds from Claims \ref{claim:lb3} and \ref{claim:lb4}, it follows that 
\[  \sum_{v \in \hat{\cE}_0} \pr{v\; |\;  c=1} > (9/10)^2 = .81 \]
Moreover, because by Claim \ref{claim:lb1} the algorithm always outputs $c=0$ when it ends in any $v \in \cE$, it follows that the algorithm incorrectly determined the value of $c$ with probability at least $.81$ conditioned on $c=1$, and therefore is incorrect with probability at least $.405 > 4/10$ which is a contradiction since $\cA$ was assumed to have success probability at least $6/10$. 

\paragraph{From the Hard Distribution to $k$-Centers.} 

To complete the proof, it suffices to demonstrate that the optimal $k$-centers cost is at most $1$ when $\bD \sim \bD_0$, and at least $R$ when $\bD \sim \bD_1$. The first case is clear, since we can choose at least one index in the pre-image of $h^{-1}(t)\subseteq [n]$ for each $t \in [k]$ to be a center. For the second case, note that conditioned on $|h^{-1}(t)| \geq 2$ for all $t \in [k]$, the resulting metric contains $k+1$ points which are pairwise-distance at least $R$ from each other. In particular, for the resulting metric, at least one point must map to a center which it is distance at least $R$ away from, and therefore the cost is at least $R$. Now since $n = \Omega( k \log k)$ with a sufficently large constant, it follows by the coupon collector's argument that with probability at least $1/1000$, we have that $|h^{-1}(t)| \geq 2$ for all $t \in [k]$. Moreover, that the $1/1000$ probability under which does not occur can be subtracted into the failure probability of $.405$ in the earlier argument, which still results in a $.404 > 4/10$ failure probability, and therefore leads and leading to the same contradiction, which completes the proof. 

\end{proof}

\begin{proposition}\label{prop:lb}
Fix any $1 \leq k < n$. Then any algorithm which, given oracle access the distance matrix $\cD \in \R^n$ of a set $X$ of $n$ points in a metric space, determines correctly with probability $2/3$ whether the optimal $k$-centers cost on $X$ is at most $1$ or at least $R$, for any value $R >1$, must make at least $\Omega(k^2)$ queries in expectation to $\cD$. The same bound holds true replacing the $k$-centers objective with $(k,z)$-clustering, for any constant $z>0$. 
\end{proposition}
\begin{proof}
By the same arguements given in Theorem \ref{thm:LBBig}, one can assume that the existence of such an algorithm that would violate the statement of the proposition implies the existence of a deterministic algorithm which always makes at most $c k^2$ queries to $\cD$ and is correct with probability $3/5$, for some arbitrarily small constant $c$. 
In what follows, we assume $n=k+1$, and for larger $n$ we will simply add $n-(k+1)$ duplicate points on top of the first point in the following distribution; note that any algorithm for the dataset with the duplicate point can be simulated, with no increase in query complexity, via access to the distance matrix on the first $k+1$ points.

The hard instance is as then as follows.
With probability $1/2$, we give as input the distance matrix $\bD \in \R^{k+1 \times k+1} $ with $\bD_{i,j} =R $ for all $i \neq j$. Note that any $k$-centers solution must have one of the points in a cluster centered at another, and therefore the optimal $k$-centers cost is at least $R$ for this instance. In the second case, the input is $\bD$ but with a single entry $\bD_{i,j} = \bD_{j,i} = 1$, where $(i,j)$ is chosen uniformly at random. Note that the result is still a valid metric space in all cases. Moreover, note that the optimal $k$-centers cost is $1$, and is obtained by choosing all points except $i$ (or alternatively except $j$). 

By the same (and in fact simplified) argument as in Claim \ref{claim:lb1}, we can assume the algorithm returns that the $k$-centers cost is at most $1$ if and only if it sees an entry with value equal to $1$. Since the algorithm is deterministic, and since the only distance other than $1$ is $R$, we can define a deterministic set $S$ of $c k^2$ indices in $\binom{k+1}{2}$ such that the adaptive algorithm would choose exactly the set $S$ if, for every query it made, it observed the distance $R$ (and therefore would return that the $k$-centers cost was at most $1$ at the end). Then in the second case, the probability that $(i,j)$ is contained in $S$ is at most $\frac{4}{c}$. Setting $c > 40$, it follows that the algorithm is incorrect with probability at least $1/2-\frac{4}{c} > 2/5$, contradicting the claimed success probability of $3/5$, and completing the proof. 
\end{proof}

\bibliography{cluster}
\end{document}